\documentclass[letterpaper,UKenglish,leqno,11pt]{article}

\usepackage[utf8]{inputenc}
\usepackage[T1]{fontenc}
\usepackage[margin=1in,letterpaper]{geometry}
\usepackage{relsize,xspace}

\usepackage{xcolor}
\usepackage{mathtools}
\usepackage{microtype}
\usepackage{amsmath}
\usepackage{amssymb}
\usepackage{mathrsfs}  
\usepackage{amsfonts}
\usepackage{hyperref}
\usepackage{bm}
\usepackage{booktabs}
\usepackage{tabularx}
\usepackage{multirow}
\usepackage[switch,displaymath]{lineno}

\usepackage{enumerate,comment}

\usepackage[inline]{enumitem}
\usepackage{latexsym}
\usepackage{authblk}
\usepackage{subfiles}
\usepackage{nicefrac}
\usepackage{tikz}
\usetikzlibrary{arrows.meta}
\usepackage[singlelinecheck=false]{caption}
\usepackage{thm-restate}

\newlength{\oldarraycolsep}

\renewcommand{\phi}{\varphi}

\newcommand{\FFF}{{\cal F}}

\newcommand{\LLL}{{\cal L}}

\newcommand{\OOO}{{\cal O}}
\newcommand{\PPP}{{\cal P}}

\newcommand{\SSS}{{\cal S}}
\newcommand{\TTT}{{\cal T}}

\newcommand{\WWW}{{\cal W}}

\newcommand{\N}{{\mathbb N}}

\newcommand{\sth}{\mathrel : }

\usepackage[amsmath,thmmarks,hyperref]{ntheorem}
\usepackage{cleveref}

\crefformat{page}{#2page~#1#3}%
\Crefformat{page}{#2Page~#1#3}%
\crefformat{equation}{#2(#1)#3}%
\Crefformat{equation}{#2(#1)#3}%
\crefformat{figure}{#2Figure~#1#3}%
\Crefformat{figure}{#2Figure~#1#3}%
\crefformat{section}{#2Section~#1#3}
\Crefformat{section}{#2Section~#1#3}
\crefformat{chapter}{#2Chapter~#1#3}
\Crefformat{chapter}{#2Chapter~#1#3}
\crefformat{chapter*}{#2Chapter~#1#3}
\Crefformat{chapter*}{#2Chapter~#1#3}
\crefformat{part}{#2Part~#1#3}
\Crefformat{part}{#2Part~#1#3}
\crefformat{enumi}{#2(#1)#3}
\Crefformat{enumi}{#2(#1)#3}

\theoremnumbering{arabic}
\theoremstyle{plain}
\theoremsymbol{}
\theorembodyfont{\itshape}
\theoremheaderfont{\normalfont\scshape}
\theoremseparator{.}

\newtheorem{theorem}{Theorem}[section]
\crefformat{theorem}{#2Theorem~#1#3}
\Crefformat{theorem}{#2Theorem~#1#3}

\newcommand{\newtheoremwithcrefformat}[2]{%
  \newtheorem{#1}[theorem]{\textsc{#2}}%
  \crefformat{#1}{##2\MakeUppercase#1~##1##3}%
  \Crefformat{#1}{##2\MakeUppercase#1~##1##3}%
}

  \newtheorem{lemma}[theorem]{Lemma}%
  \crefformat{lemma}{#2Lemma~#1#3}%
  \Crefformat{lemma}{#2Lemma~#1#3}%
  \crefformat{conjecture}{#2Conjecture~#1#3}%
  \Crefformat{conjecture}{#2Conjecture~#1#3}%
  \crefformat{proposition}{#2Proposition~#1#3}%
  \Crefformat{proposition}{#2Proposition~#1#3}%
  \newtheorem{corollary}[theorem]{Corollary}%
  \crefformat{corollary}{#2Corollary~#1#3}%
  \Crefformat{corollary}{#2Corollary~#1#3}%
\theorembodyfont{\upshape}
\newtheoremwithcrefformat{example}{Example}
\newtheoremwithcrefformat{remark}{Remark}
\newtheoremwithcrefformat{definition}{Definition}
\newtheoremwithcrefformat{notation}{Notation}
\newtheoremwithcrefformat{observation}{Observation}

\crefformat{reduction}{#2Reduction~#1#3}%
\Crefformat{reduction}{#2Reduction~#1#3}%
\theoremsymbol{\ensuremath{\square}}
\theoremheaderfont{\scshape}
\theorembodyfont{\normalfont}
\theoremsymbol{\ensuremath{\square}}

\theoremstyle{nonumberplain}

\newenvironment{cenv}{\begin{list}{}{%
      \setlength{\labelwidth}{1.5em}%
      \setlength{\leftmargin}{\labelwidth}%
      \addtolength{\leftmargin}{\labelsep}%
      \setlength{\listparindent}{0em}%
      \setlength{\topsep}{10pt}%
      \setlength{\itemsep}{5pt}%
      \setlength{\parsep}{0pt}%
    }
  }{
  \end{list}
}

\newcounter{claimcounter}
\newcounter{conditioncounter}

\newenvironment{Claim}{
  
  \refstepcounter{claimcounter}
  \begin{cenv}
  \item[{Claim \arabic{claimcounter}.}]
  }{
  \end{cenv}
}

\crefformat{claimcounter}{#2Claim~#1#3}%

\newenvironment{ClaimProof}[1][]{\noindent{%
\ifthenelse{\equal{#1}{}}{{\itshape Proof.\ }}{{\itshape #1.\ }}%
}}{\hspace*{1em}\nobreak\hfill$\dashv$\endtrivlist\addvspace{2ex plus
0.5ex minus0.1ex}}

\newenvironment{proof}[1][]
{\setcounter{claimcounter}{0}\ifthenelse{\equal{#1}{}}{\noindent\textit{Proof.
    }}{\noindent\textit{#1. }}}%
{\hspace*{1pt}\hfill$\Box$\par\bigskip}

\def\cqedsymbol{\ifmmode$\lrcorner$\else{\unskip\nobreak\hfil
\penalty50\hskip1em\null\nobreak\hfil$\lrcorner$
\parfillskip=0pt\finalhyphendemerits=0\endgraf}\fi}

\usepackage[disable]{todonotes}
\presetkeys{todonotes}{inline}{}
\presetkeys{todonotes}{fancyline, color=red!30}{}

\newcommand{\dtw}{\textit{dtw}}

\newcommand{\sBound}{\textit{Bound}}
\newcommand{\lsplit}{\textit{split}}
\newcommand{\lBound}{\textit{bound}}
\newcommand{\dtwbound}{f_{\ref{thm:directed-grid}}(\wallbound)} 
\newcommand{\wallbound}{189} 
\newcommand{\Abs}[1]{| #1 |}

\newcommand{\problemdef}[3]{
  \begin{center}
    \begin{minipage}{0.95\textwidth}
      \noindent
         \vspace{-0.5em}\colorbox{gray!17!white}{\textsc{#1}}
      
      \vspace{4pt}
      \setlength{\tabcolsep}{3pt}
      \begin{tabularx}{\textwidth}{@{}lX@{}}
        \textbf{Input:} 		& #2 \\
        \textbf{Question:} 	& #3
      \end{tabularx}
    \end{minipage}
  \end{center}
 }

\title{The Directed Disjoint Paths Problem with Congestion}

  \author[2]{Matthias Bentert
    \thanks{\texttt{bentert@tu-berlin.de}}} \author[2]{Dario
    Cavallaro\thanks{\texttt{d.cavallaro@tu-berlin.de}}} \author[2]{Amelie Heindl
    \thanks{\texttt{a.heindl@tu-berlin.de}}} \author[3,4]{Ken-ichi
    Kawarabayashi \thanks{\texttt{k\_keniti@nii.ac.jp} }} \author[2]{Stephan Kreutzer
    \thanks{\texttt{stephan.kreutzer@tu-berlin.de}}}
  \author[2]{Johannes Schr\"oder
    \thanks{\texttt{j.schroeder.1@tu-berlin.de}}}

  \affil[2]{Technical
    University Berlin, Germany} \affil[3]{National Institute of
    Informatics, Japan} \affil[4]{The university of Tokyo, Japan}
  \date{}

\newcommand{\cidp}{\textsc{$k$-Disjoint Directed Paths with Congestion~$c$}}

\colorlet{color1}{red}
\colorlet{color2}{green}
\colorlet{color3}{blue}
\colorlet{color4}{magenta}
\colorlet{color5}{brown}
\colorlet{color6}{orange}
\colorlet{color7}{violet}
\colorlet{color8}{cyan}
\colorlet{color9}{purple}

\begin{document}
\maketitle
\begin{abstract}
The classic result by Fortune, Hopcroft, and Wyllie [TCS~'80] states that the directed disjoint paths problem is NP-complete even for two pairs of terminals. 
Extending this well-known result, we show that the directed disjoint paths problem is NP-complete for any constant congestion $c \geq 1$ and~$k \geq 3c-1$ pairs of terminals. This refutes a conjecture by Giannopoulou et al. [SODA~'22], which says that the directed disjoint paths problem with congestion two is polynomial-time solvable for any constant number~$k$ of terminal pairs.
We then consider the cases that are not covered by this hardness result. The first nontrivial case is $c=2$ and $k = 3$.
Our second main result is to show that this case is polynomial-time solvable.
\end{abstract}

\newpage

\section{Introduction}

The \textsc{Disjoint Paths} problem is defined as follows. The input consists
of a graph $G$ and $k$
pairs $(s_1, t_1), \dots,$ $(s_k, t_k)$ of vertices, called the
\emph{sources} and \emph{sinks}.
For any $c\geq 1$, a \emph{solution with congestion $c$}
is a set $P_1, \dots, P_k$ of paths such that $P_i$ links $s_i$ and
$t_i$ and every vertex $v\in V(G)$ is contained in at most $c$ of the
paths. For digraphs, the paths $P_i$ are
required to be directed from $s_i$ to $t_i$.
In case $c=1$ ($c=2$, respectively) we call $P_1, \dots,P_k$ an
\emph{integral} (\emph{half-integral}, respectively) solution.

The \textsc{Disjoint Paths} problem is one of the most fundamental problems in graph algorithms and combinatorial optimization. 
According to Menger's theorem, it can be solved in polynomial time on undirected and directed graphs if we only ask for a set of disjoint paths, each having one end in~$S=\{s_1, \dots, s_k\}$ and the other in~$T=\{t_1, \dots, t_k\}$. But the situation changes completely if we require that the paths connect each source $s_i$ to its corresponding target $t_i$; it is well-known that this problem is NP-complete for vertex-disjoint paths on directed and undirected graphs \cite{Karp75,FHW80}. 

For undirected graphs, Robertson and Seymour developed an algorithm for the~\textsc{$k$-Disjoint Paths} problem, which runs in time~$O(n^3)$ for any fixed number~$k$ of terminal pairs \cite{GMXIII}. In the terminology of parameterized complexity, they showed that the problem is fixed-parameter tractable parameterized by the number~$k$ of terminal pairs. The complexity has subsequently been improved to quadratic time in \cite{kkr2012} and very recently improved to almost linear time in \cite{2024minor}. Robertson and Seymour developed the algorithm as part of their celebrated series of papers on graph minors. The correctness proof is still long and difficult and relies on large parts of the graph minors series (even with the new unique linkage theorem in \cite{unique}). Still, the algorithm is much less complicated and essentially facilitates a reduction rule that reduces any input instance to an equivalent instance of bounded tree width, which can be solved quickly.

For directed graphs, the~\textsc{$k$-Disjoint Paths} problem is much more difficult. As shown by Fortune, Hopcroft, and Wyllie  \cite{FHW80} in 1980, the problem is NP-complete already for~$k=2$ terminal pairs. This implies that it is not fixed-parameter tractable (parameterized by $k$) and not even in the class XP (under the usual complexity theoretic assumptions that we tacitly assume throughout the introduction).

Therefore, 
to attack the \textsc{$k$-Disjoint Paths} problem for general digraphs 
we need to relax the notion of ``disjoint.'' Let us first look at the undirected case. The natural multicommodity flow relaxation of 
the disjoint paths problem provides strong information for designing approximation algorithms when congestion is at least two. 
In the multicommodity flow problem, one wants to assign non-negative real values to paths connecting given terminal pairs, such that for each terminal pair the sum of the values of the paths connecting its elements is at least 1 and for each vertex in the graph the sum of the values of the paths containing it is at most 1.
Finding disjoint paths is equivalent to solving the multicommodity flow problem with the restriction to only use the integers 0 and 1 as values.
The problem relaxation leads to the integrality gap, as using non-integer values in the approximation  permits results differing from the 0-1-solution (see \cite{frank1990packing} for more details).
Half-integral paths can be regarded as the counterpart to multicommodity flows with half-integer values.
By imposing that each 
vertex is used in at most two of the paths, one can change the global structure of the routing problem. 
For example, Kleinberg \cite{kleinberg} uses a grid minor as a ``crossbar'' for the half-integral disjoint paths problem, which leads to a faster algorithm (even faster algorithm in \cite{kawarabayashi2009nearly}) with a relatively short proof for the correctness of the algorithm, while Robertson and Seymour's algorithm for the disjoint paths problem \cite{GMXIII} needs a big clique minor as a crossbar, which in turn requires all of the graph minors series. 
So allowing half-integrality 
greatly reduces the computational complexity, as well as the technical difficulties, of the disjoint paths problem. This shares 
similarities with the fact that good characterizations are known for the multicommodity flow problem with half-integer flow values for which the 
corresponding disjoint paths problem is intractable 
(see \cite{frank1990packing, middendorf1993complexity} for more details).

Taking the state-of-the-art literature on undirected graphs into account, when we relax ``disjoint'' to ``half-integral,'' it is quite natural to expect that the same phenomena could apply to the directed case, i.e.,~that, the \textsc{Directed Disjoint Paths} problem becomes tractable once we allow half-integral paths. This is stated as an open problem in \cite[Problem 9.5.8]{bang2018classes}. Indeed, that was a common assumption in the field, with several papers proving special cases of the problem or closely related results. 

In \cite{GiannopoulouKKK2022}, the authors consider a related setting for digraphs and consider the following problem.
Given a digraph and $k$ pairs of terminals, either find $k$ half-integral disjoint paths, each connecting the respective pair of terminals or determine that there are no $k$ disjoint paths connecting these pairs of terminals.
This can be regarded as a relaxation of the \textsc{$k$-Half-Integral Disjoint Paths} problem, with the additional permission to not return a half-integral solution as long as no completely disjoint solution exists.
They obtain a polynomial time algorithm for this weaker setting by using a directed grid minor as a ``crossbar'' to route the paths.
Upon having achieved a tractability improvement by allowing congestion, they conjectured that a polynomial time algorithm also exists for the more general \textsc{$k$-Half-Integral Disjoint Paths} problem (for any fixed number $k$ of terminals)~\cite[Conjecture 1.5]{GiannopoulouKKK2022}.

In \cite{EdwarsMW2017}, this conjecture was confirmed for highly connected digraphs. More precisely, it was shown that there is an absolute constant $d$ such that for each $k \geq 2$, there exists a function~$L(k)$ such that the \textsc{$k$-Half-Integral Disjoint Paths} problem is solvable in time $O(|V(G)|^d)$ for a strongly $L(k)$-connected directed graph $G$.

  This result was later improved by Campos et al. in \cite{CamposCLS2023}, who reduced the bound on the strong connectivity required for half-integral linkages to exist. In the same paper, the authors also introduce new types of separators geared towards proving tractability results for half-integral disjoint paths problems in directed graphs. 
  If congestion $8$ is allowed, the bounds have further been improved by Masa{\v r}{\'i}k et al. \cite{MasarikPRS2022}, which then have been improved even further by Lopes et al. in \cite{lopes2024constantcongestionlinkagespolynomially}.

  The previous results settle the tractability of the half-integral and congestion $c$ disjoint paths problems on very highly connected digraphs. The other extreme case was settled by Amiri et al. \cite{AmiriKMR19} who showed that the \textsc{$k$-Directed Paths} problem can be solved in time $n^{O(d)}$ on acyclic digraphs if congestion $k-d$ is allowed. This aligns with the known algorithm for the \textsc{$k$-Disjoint Paths} problem on acyclic digraphs running in time $n^{O(k)}$. This algorithm for DAGs was generalized by Johnson et al. \cite{JohnsonRST2001}. They showed that for each fixed $k \geq 1$, the \textsc{$k$-Disjoint Paths} problem can be solved in polynomial time on any class of digraphs of bounded directed tree-width. (See \cref{sec:dtw} for a definition.)

  Low congestion routing in directed graphs has also been studied in other contexts, see e.g.,~\cite{ChekuriE2015, ChekuriEP2016} and the references therein.

  Despite this considerable effort to prove that the \textsc{$k$-Half-Integral Disjoint Paths} problem can be solved in polynomial time for any fixed number $k$ of terminal pairs, no such algorithms have been found to date. In fact, despite some attempts, a polynomial time algorithm is not even known for the simplest non-trivial case of $k=3$ terminal pairs (the problem is trivial for $k\leq 2$).
  
  In this paper, we settle the problem for almost all $k \in \N$. 

  Our first main contribution is to show that the \textsc{$k$-Half-Integral Disjoint Paths} problem is NP-complete for all $k \geq 5$. This result is somewhat surprising as it was widely believed in the community that the problem would be tractable (and as mentioned above, there are a lot of partial results).
  In fact we show the following stronger result, which also generalizes the above-mentioned result by Fortune, Hopcroft, and Wyllie  \cite{FHW80}.

\begin{restatable}{theorem}{nphardness}
    \label{thm:nphard}
    \cidp{} is NP-complete for any constant~$c \geq 1$ and any~$k \geq 3c-1$.
\end{restatable}

In particular, our result refutes the conjecture in \cite{GiannopoulouKKK2022}.
As mentioned above, it was already known that the complexity of the fully \textsc{Disjoint Paths} problem for digraphs is totally different from the complexity of the undirected case. Our result shows that this is also the case even for congestion $c \geq 2$.
Further, \cref{thm:nphard} settles an open question raised in \cite[Section 5]{lopes2022relaxation} regarding the computational complexity of the \textsc{Disjoint Enough Directed Paths} (short \textsc{DEDP}) problem.

However, this strong hardness result does not cover the case of congestion $c=2$ and $k=3$ terminal pairs. Indeed, this case is polynomial-time solvable. 
This is the second main result of this paper.
\begin{restatable}{theorem}{main}
    \label{thm:main}
    The \textsc{$3$-Half-Integral Directed Paths} problem is solvable in
  polynomial time. 
\end{restatable}

Note that the cases $k \leq c$ are trivial (each pair can be solved independently), so this result settles the first non-trivial case with $c > 1$. 

\paragraph{Proof overview. } 
First, we give a high-level overview of the algorithm for solving the \textsc{$3$-Half-Integral Directed Paths} problem. 
We start with a standard approach to solving this kind of problems: if the directed tree-width of the input digraph $G$ is bounded by some fixed constant, then we can solve the problem directly by applying an  adaptation of the disjoint paths algorithm by Johnson et al. \cite{JohnsonRST2001}; see \cref{sec:dtw} for details. Otherwise, by the directed grid theorem \cite{KawarabayashiK2015,HatzelKMM2024a}, $G$ contains a large cylindrical wall that can be used for routing provided it is sufficiently connected to the terminals. See \cref{sec:dtw} for details. This approach has been followed in several other papers before.

What remains is to handle the case where $G$ contains a large enough wall $W$ that is not sufficiently linked to the terminals. This is by far the hardest case, as we explain now. Menger's theorem implies that if there are no three disjoint paths from $\bar s := (s_1, s_2, s_3)$ to $W$, then there is a separation $(A_1, B_1)$ in $G$ of order $\leq 2$ such that $A_1$ contains the sources $s_i$ whereas the majority of $W$ is contained in $B_1$. Similarly we get a separation $(A_2, B_2)$ of order $\leq 2$ with $t_i$ in $A_2$ and the majority of $W$ in $B_2$. But these are \emph{directed separations}. That is, while every path from some $s_i$ to $W$ must go through the at most two vertices in $A_1 \cap B_1$, we have no control over how the paths from $W$ to the targets $t_i$ cross from $B_1$ to $A_1$. This is completely different from the undirected case. If $G$ were undirected, we would get a single separation of order $\leq 5$ that simultaneously separates $W$ from the $s_i$ and the $t_i$. But here, we get two different separations $(A_1, B_1)$ and $(A_2, B_2)$, which can cross each other and interact in a complicated way. Controlling the interaction between the paths of a solution and crossing separations in $G$ is a challenging task, and as our hardness result for $k\geq 5$ proves, this is impossible already for five paths and thus separations of order~$\leq 4$.

But for separations of order at most two, this can be done, as we will show next. The key here is that we exploit a separation of order smaller than three by locally routing the paths which cross in its separator together, which allows us to simulate being in the case for 2 paths. Explicitly, when given a separation $(A, B)$ of order at most $2$ with no cross edge from $A$ to $B$ such that $s_i, t_i \in A$, for all $1 \leq i \leq 3$, and $B \setminus A$ non-empty, i.e., containing our wall $W$. Then every path between the terminals that contains a vertex of $B \setminus A$ must enter $B$ through the at most $2$  vertices in $A \cap B$ and then leave $B$ using some edge from $B$ to $A$. In \cref{sec:reduction}, we will show that we can compute all possible ways a solution can take inside $B$ in polynomial time. Unfortunately, it turns out that, given a solution~$\LLL=\{P_1,P_2,P_3\}$ to the instance, the interactions between~$\LLL$ and~$W$ are not as predictable as anticipated. The well-known irrelevant vertex technique due to Robertson and Seymour \cite{GMXXI}, which exploits the fact that one can find vertices in a large enough wall that are not relevant to the existence of a solution, cannot be easily transferred to our setting. In fact, for~$k=4$ one can construct feasible instances on graphs~$G$ containing arbitrarily large walls~$W$ such that deleting \emph{any} vertex in~$V(W)$ renders the instance infeasible, i.e., \emph{no} vertex of the wall is irrelevant \cite[Theorem 6.2.1]{milanidigraph}. In the case of~$k=3$ together with a~$2$-separation~$(A,B)$ separating the wall from the sources as above, a solution~$\LLL$ may induce~$4$ maximal sub-paths in~$G[B]$ starting in~$A\cap B$, which, a priori, may use all the vertices in~$V(W)$. It turns out however that we can solve this special case nonetheless---we do so in \cref{sec:reduction}---marking an important ingredient in the proof of the polynomial-time algorithm for \textsc{$3$-Half-Integral Directed Paths}, and providing some insight on a special case for~$k = 4$ that may be of relevance for future work. 

Once this is computed, we no longer need $B$ as we can always look up possible connections through $B$ in our table. Thus, we can now focus on the possible connections a solution can take inside $A$. If $G[A]$ has bounded directed tree-width, we can use the algorithm from \cref{sec:dtw}. Otherwise, we get another wall $W_2$ and repeat the process. 

The main challenge here is to handle separations that cross each other (see \cref{sec:uncross} for definitions). 
Fortunately, for separations of order at most~$2$ this turns out to be possible. In \cref{sec:constructing-dtw} we show that we can always uncross the two separations $(A_1, B_1)$ and $(A_2, B_2)$. Continuing in this way we split off more and more walls from the graph until we are left with a digraph of bounded directed tree-width. At this point, we apply our algorithm from \cref{sec:dtw} that solves instances of bounded directed tree-width with additional look-up tables for the separations we split off as discussed above. 

\medskip

We now give a brief overview of our NP-hardness result for all $k \geq 5$. Fortune et al. \cite{FHW80} proved that the $2$-disjoint paths problem is NP-complete on directed graphs for $k=2$ terminal pairs. In their proof, they construct a reduction from 3-SAT to the $2$-disjoint paths problem. Given a propositional formula in $3$-CNF, they construct a simple gadget for each variable. The first path of a solution must cross each gadget and, in this way, choose a variable assignment. To verify that this assignment satisfies the formula their proof uses a cleverly constructed gadget for each clause $C$ and each variable $X$ contained in $C$. The gadget can be crossed in two completely separate ways, but using the second path of the solution, they make sure that only one of the two options is possible for each gadget. Thus, the gadgets act as a switch, and choosing the wrong assignment for a variable $X$ disables routing through the clause $C$. This gadget is the core of their argument. 

We want to use a similar approach to show hardness for the half-integral paths problem on $k \geq 5$ terminal pairs. The main problem is that since we allow vertices to be used by two different paths, the switching gadget of \cite{FHW80} no longer works as the first path can simply shortcut without selecting any variable assignment or verifying any such choice. 
In our proof, we use the same general setup and also reduce from 3-SAT. The main difficulty is to design a switching gadget that allows the paths selecting variables to cross in one of two ways without making it possible for the path to take a shortcut. So, the core of the hardness result is a careful design of a switching gadget that allows us to reduce from 3-SAT.

On a high level, we want to ensure that a collection of~$c$ paths have two completely different options to route through the new switch gadget but all~$c$ paths should make the same decision.
Amusingly, we achieve this by first forcing them to each take a different path.
At the heart of our reduction lies the following structure.
Assume that there are three sources~$s_1,s_2,s_3$ and four sinks~$t_1,t'_1,t_2,t_3$.
Our goal will be to either route~$c$ paths from~$s_1$ to~$t_1$ or to route~$c$ paths from~$s_1$ to~$t'_1$ (this will encode our choice of a variable assignment).
Moreover, assume that~$2c-1$ paths have to be routed from~$s_2$ or~$s_3$ to~$t_2$ or~$t_3$ (how many paths connect e.g.,~$s_2$ and~$t_2$ does not matter).
Let there now be~$2c$ additional vertices~$u_1,u_2,\ldots,u_c,v_1,v_2,\ldots,v_c$ and edges~$(u_i,u_{i+1})$ and~$(v_i,v_{i+1})$ for each~$i \in [c-1]$.
We also add the edges~$(s_2,u_1)$, $(s_2,v_1)$, $(s_3,u_1)$, $(s_3,v_1)$, $(u_c,t_2)$, $(v_c,t_2)$, $(u_c,t_3)$, and~$(v_c,t_3)$.
Finally, for each~$i \in [c]$, we add the edges~$(s_1,u_i)$, $(s_1,v_i)$, $(u_i,t_1)$, and~$(v_i,t'_1)$.
See \cref{fig:tinyswitch} for an illustration.
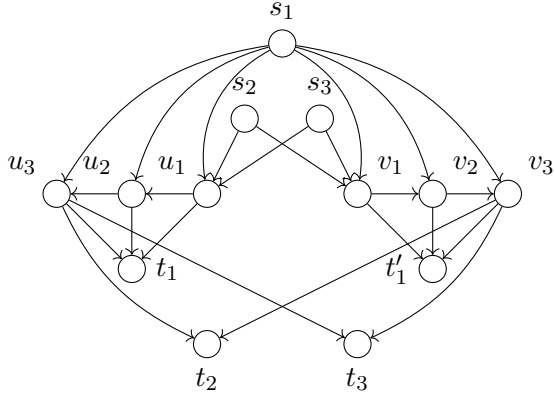
\begin{figure}[t!]
    \centering
    \begin{tikzpicture}
        \foreach \i in {1,2,3}{
            \node[circle,draw,label=above left:$u_{\i}$] at(-\i,0) (u\i) {};
            \node[circle,draw,label=above right:$v_{\i}$] at(\i,0) (v\i) {};
        }
        \node[circle,draw,label=$s_1$] at(0,2) (s1) {};
        \node[circle,draw,label=right:$t_1$] at(-2,-1) (t1) {};
        \node[circle,draw,label=left:$t'_1$] at(2,-1) (tt1) {};
        \node[circle,draw,label=$s_2$] at(-.5,1) (s2) {};
        \node[circle,draw,label=$s_3$] at(.5,1) (s3) {};
        \node[circle,draw,label=below:$t_2$] at (-1,-2) (t3) {};
        \node[circle,draw,label=below:$t_3$] at (1,-2) (t4) {};
        \foreach \i in {1,2,3}{
            \draw[->] (u\i) to (t1);
            \draw[->] (v\i) to (tt1);
            \draw[->,bend right=40-5*\i] (s1) to (u\i);
            \draw[->,bend left=40-5*\i] (s1) to (v\i);
        }
        \foreach \i in {2,3}{
            \pgfmathtruncatemacro\j{\i-1};
            \draw[->] (u\j) to (u\i);
            \draw[->] (v\j) to (v\i);
        }
        \draw[->] (s2) to (u1);
        \draw[->] (s2) to (v1);
        \draw[->] (s3) to (u1);
        \draw[->] (s3) to (v1);
        \draw[->,bend right=20] (u3) to (t3);
        \draw[->] (v3) to (t3);
        \draw[->] (u3) to (t4);
        \draw[->,bend left=20] (v3) to (t4);
    \end{tikzpicture}
    \caption{An illustration of the main idea to force~$c=3$ paths to take a unanimous decision.}
    \label{fig:tinyswitch}
\end{figure}
Note that from the~$2c-1$ paths starting from~$s_2$ and~$s_3$, either~$c$ paths route through all nodes~$u_i$ and~$c-1$ paths route through all nodes~$v_i$ or vice versa.
Assume that~$c$ paths route through~$u_i$ and~$c-1$ route through~$v_i$.
Then, all paths starting in~$s_1$ have to go to a distinct vertex~$v_i$ and from there they all continue to~$t'_1$.
Note that reducing the number of paths in either of the two sets by one breaks the construction.
If there are only~$2c-2$ paths starting in~$s_2$ or~$s_3$, then~$c-1$ paths can go to~$u_1$ and~$c-1$ paths can go towards~$v_1$.
Hence, some paths starting in~$s_1$ can go towards~$t_1$ and others can go towards~$t'_1$.
If the number of paths starting in~$s_1$ is at most~$c-1$, then even if all paths route to the same node (say~$t_1$), this does not block all other paths from using the node and hence does not provide the switching property we desire.
While the above intuitive idea is relatively simple, the actual gadget requires some more care.
In particular, we have to ensure that all paths starting in~$s_2$ or~$s_3$ have to go to~$t_2$ or~$t_3$ and not to~$t_1$ or~$t'_1$.
We present the details of the construction in \cref{sec:hardness}.

\paragraph{Open problems. }
From our results, the only open case for half-integral paths is the case with $k=4$ terminal pairs. It is conceivable that $k=4$ terminal pairs can also be decided in polynomial time; in principle, one could try the same approach outlined above. But computing the possible routings on the side of a separation not containing the terminals and uncrossing order three separations becomes considerably more complicated, and we leave this to future research.

On the hardness side, the remaining problem is whether or not $k \geq 3c-1$ is best possible. If the $k=4$ case can be solved in polynomial time, then this implies that at least  for $c=2$,  \cref{thm:nphard} is best possible. 

Another direction one can consider for the \textsc{$k$-Disjoint Paths} problem for digraphs is restricting the input to some graph family. This direction has already been considered; for example, when the input digraph is planar, the \textsc{$k$-Disjoint Paths} problem is tractable (even fixed-parameter tractable), see \cite{planardp, lex}. Our result, Theorem \ref{thm:nphard}, suggests this is the right direction, even for congestion $c$. So we expect that such results will come in the near future.

\section{Preliminaries}
\label{sec:prelims}
\paragraph{Graph Theory.} Throughout this paper we use well-established notation and results for undirected and directed graphs (or digraphs) recalling the most relevant definitions to this paper below; see \cite{bondy2008graph} and \cite{bang2018classes} respectively. To summarise, we denote undirected and directed graphs by~$G$, and write~$V(G)$ and~$E(G)$ for their vertex and edge-set respectively. We work with standard graph classes such as (directed) paths, walks, and cycles. For simplicity we may write paths~$P\subseteq G$ via~$P=(v_1,\ldots,v_k)$ for~$k = \Abs{V(P)}$ highlighting the associated order on its vertices, and similarly for walks and cycles. We may add the edges to the tuple if we want to highlight the exact edges used by a path, walk or cycle, i.e,~$P=(v_1,e_1,v_2,\ldots,v_{k-1},e_{k-1},v_k)$ where~$e_i=(v_i,v_{i+1})\in E(G)$ for~$1 \leq i \leq k-1$. We call a directed path starting in~$v_1$ and ending in~$v_k$ a~$(v_1,v_k)$-path and if~$v_1 \in S \subseteq V(G)$ and~$v_k \in T \subseteq V(G)$ we call it an~$(S,T)$-path. We write~$G-U$ to denote the graph induced by~$V(G) \setminus U$. Recall that a digraph~$G$ is \emph{strongly connected} if it is either a single vertex, or if for every pair of distinct vertices~$x,y \in V(G)$ there exists a closed walk in~$G$ visiting both vertices, and it is strongly~$k$-connected if deleting~$<k$ vertices in~$G$ leaves the graph strongly connected.

\begin{definition}[Disjoint paths]
    Let~$G$ be a digraph and~$P_1,P_2 \subseteq V(G)$ be two directed paths. We say that~$P_1,P_2$ are \emph{(vertex) disjoint} if~$V(P_1) \cap V(P_2) = \emptyset$. We say that they are \emph{internally disjoint} if for every internal vertex~$v \in V(P_i)$ it holds~$v \notin V(P_1) \cap V(P_2)$ for~$i =1,2$.
\end{definition}

We define directed separations as follows.

\begin{definition}[Directed Separations and Separators]\label{def:separation}
  Let $G$ be a digraph and let $X, Y \subseteq V(G)$. A \emph{cross edge} for
  the pair $(X, Y)$ is an edge $e\in E(G)$ with tail in $Y \setminus X$ and head in $X \setminus
  Y$ or tail in $X \setminus Y$ and head in $Y \setminus X$. We say that $e$ crosses
  \emph{from $X$ to $Y$} if its tail is in $X$, otherwise it crosses
  \emph{from $Y$ to $X$}. 

  A \emph{(directed) separation} in $G$ is a pair $S=(A, B)$ of sets $A, B \subseteq
  V(G)$ such that $A \cup B = V(G)$ and there are no cross edges from $A$
  to $B$ or no cross edges from $B$ to $A$ in $G$. We call $A$ and $B$ the \emph{sides} of $S$.

  We define $\partial^+(S) := A$ and $\partial^-(S) := B$  if there are no cross edges from $B$ to $A$
  and otherwise set $\partial^+(S) := B$ and $\partial^-(S) := A$.

  The \emph{order} of $S$, denoted by $|S|$, is defined as $|A \cap B|$. 
  We call $A \cap B$ the \emph{separator} of $S$. We call separators of size $1$ \emph{cut vertices}.

   Given two sets~$U,W \subseteq V(G)$ such that~$U \subseteq A$ and~$W \subseteq B$ we say that~$(A,B)$ \emph{separates~$U$ from~$W$}. If additionally~$U \subseteq A\setminus B$ and~$W \subseteq B \setminus A$ we say that~$(A,B)$ \emph{internally} separates~$U$ from~$W$.
\end{definition}

For separations $S = (A, B)$ for which there are no cross edges in
either direction the definition of $\partial^+(S)$ and $\partial^-(S)$ is ambiguous.
In these cases we choose any of the two possibilities but in a
consistent way, that is so that $\partial^+((A, B)) = \partial^+((B, A))$.

There is a well known duality between separators and disjoint paths due to Menger.

\begin{theorem}[Menger's Theorem]\label{thm:Menger}
    Let~$G$ be a digraph,~$S,T \subseteq V(G)$, and~$k \in \N$. Then, either there exist~$k$ disjoint~$(S,T)$-paths or there exists a separation~$(A,B)$ of order at most $k-1$ that separates~$S$ from~$T$.

    Similarly, either there exist~$k$ internally disjoint~$(S,T)$-paths or there exists a separation~$(A,B)$ of order~ at most $k-1$ that internally separates~$S$ from~$T$.
\end{theorem}

Finally, we define the objects of interest for this paper: \emph{half-integral paths}.

\begin{definition}[Disjoint paths up-to congestion]
Let~$G$ be a digraph and~$c,k \in \N$. Let $P_1,\ldots,P_k$ be paths in~$G$. We say that the paths have \emph{congestion~$c$} if for every~$v \in V(G)$ there are at most~$c$ paths among~$\{P_1,\ldots,P_k\}$ that visit~$v$. If~$c=2$ we say that the paths~$P_1,\ldots,P_k$ are \emph{half-integral}.
\end{definition}

We call a set~$\{P_1,\ldots,P_k\}$ of disjoint paths a \emph{linkage} and call~$k$ its order. Given two sets~$S,T \subseteq V(G)$, an \emph{$(S,T)$-linkage}~$\LLL$ is a collection of~$k$ internally disjoint paths~$\LLL=\{P_1,\ldots,P_k\}$ in~$G$ such that~$P_i$ is an~$(S,T)$-path for~$1 \leq i \leq k$. \emph{Half-integral linkages} and \emph{$k$-linkages with congestion~$c$} are defined analogously with the addition that we only require the set of paths to be half-integral or to have congestion~$c$ respectively. Given any type of linkage~$\LLL$ we call the collection of ends of the paths the \emph{ends of the linkage}.
The following is the problem of interest for the rest of this paper; let~$k,c \in N$.

\problemdef{$k$-Disjoint Directed Paths with Congestion~$c$}{A digraph~$G$, tuples~$\bar{s}=(s_1,\ldots,s_k)$,~$\bar{t}=(t_1,\ldots,t_k)$ with~$s_i,t_i \in V(G)$ for~$1 \leq i \leq k$.}{Does there exist a~$k$-linkage~$\LLL=\{P_1,\ldots,P_k\}$ with congestion~$c$ such that~$P_i$ is an~$(s_i,t_i)$-path?}

If~$c=2$ we refer to the problem as the \textsc{$k$-Half-Integral Directed Paths} problem.

We refer to the tuples~$\bar{s}=(s_1,\ldots,s_k)$,~$\bar{t}=(t_1,\ldots,t_k)$ of an instance of the~\textsc{$k$-Disjoint Directed Paths with Congestion $c$} problem as \emph{terminal pairs}. For the sake of readability we will write~$\bar{s} \subseteq V(G)$ to mean some tuple~$\bar{s}=(s_1,\ldots,s_\ell)$ of arbitrary but fixed size~$\ell \in \N$ (unless stated otherwise), where each~$s_i \in V(G)$ for~$1 \leq i \leq \ell$. Also given two tuples~$\bar{s},\bar{t} \subseteq V(G)$ we write~$\bar{s} \cup \bar{t} \subseteq V(G)$ to mean the obvious. 

We may assume an instance~$(G,\bar{s},\bar{t})$ of the \textsc{$k$-Disjoint Directed Paths with Congestion~$c$} problem to respect the following \emph{normal form} by adding one vertex to the graph for each terminal.

\begin{definition}[Normal form]\label{def:normalform}
  Let~$k,c \in \N$. An instance $(G, \bar{s}, \bar{t})$ of the
  \textsc{$k$-Disjoint Directed Paths with Congestion~$c$} problem is in \emph{normal form} if, every~$s \in \bar{s}$ has exactly one outgoing and no incoming edge and
  every~$t \in \bar{t}$ has exactly one incoming and no outgoing edge.
\end{definition}

\paragraph{Computational Complexity.} We use standard notation, definitions and results regarding computational complexity (parameterized and unparameterized) following \cite{cygan2015parameterized} unless stated otherwise. We assume the reader to be familiar with basic computational complexity concepts such as the classes~$P$ and~$NP$, as well as the concept of~$NP$-completeness. Most notably, a parameterized problem is given a parameter~$k \in \N$ as part of the input and when we say that an algorithm for a parameterized problem runs in \emph{$fpt$-time}, we mean that on input~$(I,k)$ where~$k$ is the parameter, the algorithm runs in time~$f(k)\Abs{I}^c$ for some constant~$c \in \N$ and some function~$f: \N \to \N$. In particular, the algorithm runs in polynomial time for any fixed value of~$k$; we call such an algorithm an \emph{fpt-algorithm}. We refer the reader to \cite{cygan2015parameterized} for details.

\section{Reductions and Special Cases}
\label{sec:reduction}

The goal of this section is to prove several special cases of the \textsc{$3$-Half-Integral Directed Paths} problem which we will later use to reduce the general problem to the case of bounded directed tree-width. 

To unify the presentation we will work with the following technical definition of a variant of the problem. Let $k, s, t, c\in \N$ be integers.
\problemdef{$(k, s, t, c)$-Disjoint Paths}{A directed graph $G$, distinct vertices $x_1, \dots, x_s,
  y_1, \dots, y_t \in V(G)$, and two surjective functions $\sigma \sth [k] \rightarrow \{
  x_1, \dots, x_s\}$ and $\tau \sth [k] \rightarrow \{y_1, \dots, y_t\}$.}{

  Does there exist a~$k$-linkage~$\LLL=\{P_1, \dots, P_k\}$ with congestion~$c$ in~$G$ such that for all $1 \leq i \leq k$ the path $P_i$ is a $(\sigma(i),\tau(i))$-path?}

We refer to the problem as~$(k, s, t, c)$-DP for short. Note that for~$s=t=k$ and~$\sigma$ and~$\tau$ being the identity the above problem is exactly the \textsc{$k$-Disjoint Directed Paths with congestion~$c$} problem.
\smallskip

We start with a discussion of two special cases of the problem which will be used extensively later on. To this extent we define the following.

\begin{definition}\label{def:sequence-cut-vertices}
  Let $G$ be a digraph and $s, t_1, t_2 \in V(G)$ such that 
  $s \not= t_1, t_2$ but $t_1$ and $t_2$ may be equal.
  A
  \emph{sequence~$\bar c$ of cut vertices between $s$ and $t_1, t_2$} is a
  sequence $\bar c \coloneqq c_0 , c_1, \dots, c_k$ of vertices such that~$c_0 = s$ and for all
  $0 \leq j \leq k-1$ there are internally vertex-disjoint paths
  $P_j^1, P_j^2$ from $c_j$ to $c_{j+1}$, every path from $c_j$ to
  $\{t_1, t_2\}$ contains $c_{j+1}$, and either $c_k = t_1 = t_2$ or
  there are two internally vertex disjoint paths $P_k^1, P_k^2$ from
  $c_k$ to $t_1$ and $t_2$, respectively. In case that $c_k=t_1=t_2$ we set $P_k^1=P_k^2=c_k$.

  We define~$P^i_G(\bar c) \coloneqq \bigcup_{0 \leq j \leq k}P_j^i$ for~$i=1,2$ after a fixed but arbitrary choice of  paths $P_j^i$
\end{definition}

Let $\bar c := c_0, \dots, c_k$ be a sequence of cut vertices between $s$ and
$t_1, t_2$. We abuse notation and write~$V(\bar{c}) \coloneqq \{c_0,\ldots,c_k\}$. Note that~$V(P^1_G(\bar c)) \cap V(P^2_G(\bar c)) = V(\bar c)$. 

For $0 \leq i < k$ we define $R_i(\bar c)$ as the set of all
vertices in $G$ that lie on a path from $c_i$ to $c_{i+1}$. We define
$R_k(\bar c)$ to be the set of vertices on a path from $c_k$ to $t_1$
or on a path from $c_k$ to $t_2$. It follows from the definition that
for $0 \leq i < j \leq k$ the sets $R_i(\bar c)$ and $R_j(\bar c)$ are
disjoint except possibly for a shared cut vertex $c_j$ in case
$j=i+1$. Furthermore, the paths $P_j^1, P_j^2$ are contained in
$R_j(\bar c)$, for all $0 \leq j \leq k$. 

The following lemma follows easily from a repeated application of Menger's \cref{thm:Menger}.

\begin{lemma}\label{lem:sequence-cut-vertices}
  Let $G$ be a digraph and let $s, t_1, t_2$ such that $s \not= t_1,
  t_2$. If $G$ contains paths from $s$ to $t_1$ and from $s$ to $t_2$
  then there is a unique sequence of cut vertices between
  $s$ and $t_1, t_2$ and this sequence can be computed in polynomial time.
\end{lemma}

\begin{proof}
    We devise a simple algorithm by initializing the cut sequence as $c= s$.
    We then apply the internally disjoint version of Menger's \cref{thm:Menger} for $k=2$ with $S = \{s\}$ and $T=\{t_1,t_2\}$. If we find a separator of size $1$, we add it to the cut sequence and repeat Menger's Theorem with the newly found separator as $S$. If we find $2$ paths, we exit the algorithm with the current cut sequence. Note that Menger's theorem can be leveraged to find the separator closest to either set of terminals and we always search for the separator closest to $S$. Through this we find every separator of order $1$ between $s$ and $t_1,t_2$ and thus construct the unique cut sequence.
\end{proof}

We are now ready to solve the first special case of $(k, s, t, c)$-DP.

\begin{lemma}\label{lem:two-source-dp}
  For $k=3, s=2, c=2$, and $t \in \{2, 3\}$, the $(k, s, t, c)$-DP problem can be
  solved in polynomial time.  
\end{lemma}
\begin{proof}
Let $G$,~$x_1, x_2, y_1, y_2, y_3 \in V(G)$, and the functions $\sigma$ and $\tau$ be
  given. We require that $x_1 \not= x_2$ and that $y_1, y_2, y_3$ are not all equal, as in these cases the instance has no solution. In case $t=3$, the vertices $y_1, y_2, y_3$ are pairwise distinct, if $t=2$, then two of them are equal. Without loss of generality we assume that $\sigma(1) = \sigma(2) = x_1$ and $\sigma(3) = x_2$
  and $\tau(i) = y_i$, for all $1 \leq i \leq 3$.

  If for some $1 \leq i \leq 3$, there is no path in $G$ from $\sigma(i)$ to
  $\tau(i)$, then the problem does not have a solution and we can reject
  the input. Thus we may assume that $\tau(i)$ is reachable
  from $\sigma(i)$, for all $1 \leq i \leq 3$.
  
  \smallskip\noindent\textit{Case 1. } Suppose, there are two internally vertex-disjoint paths $P_1$ and $P_2$ such that $P_i$ links $x_1$ to $y_i$, for $i=1,2$. 
  Let $P_3$ be a path in $G - x_1$ from $x_2$ to $y_3$. If $y_1 = y_2$ then we additionally require that $P_3$ does not contain $y_1$. If no such
    path exists then again the problem has no solution and the
    algorithm stops and rejects the input. Otherwise, $\{P_1, P_2, P_3\}$
    forms a half-integral solution which the algorithm returns. 

    \smallskip\noindent\textit{Case 2. } Now suppose that there are no
    two internally vertex disjoint paths from $x_1$ to $\{y_1, y_2\}$.
    We apply~\cref{lem:sequence-cut-vertices} to obtain a sequence
    $\bar c := c_0, \dots, c_d$ of cut vertices between $x_1$ and
    $\{y_1, y_2\}$ and the corresponding paths $P_1 := P_G^1(\bar c)$ and $P_2 := P_G^2(\bar c)$ from $x_1$ to $y_1$ and $y_2$, respectively. Observe that if $y_1 = y_2$ then $c_d = y_1$.

If there is a path $P_3$ from $x_2$ to $y_3$ in $G - \{ c_0, \dots,
    c_d\}$, then $\{P_1, P_2, P_3\}$ forms a half-integral solution that the algorithm returns. To see this note that by definition of cut sequences $V(P_1) \cap V(P_2) = \{ c_0, \dots, c_d\}$ and as~$P_3$ is disjoint from $\{ c_0, \dots, c_d\}$, no vertex in $G$ can
    appear in all three paths.

    Otherwise, if there is no such path $P_3$, then the instance
    has no solution. For, by definition of a sequence of cut vertices, in any solution witnessed by some linkage $\{P_1', P_2', P_3'\}$ the paths $P_1', P_2'$
    linking $x_1$ to $y_1$ and $y_2$ respectively must both contain all vertices
    in $\{ c_0, \dots, c_d\}$. Thus $P_3'$ cannot contain any vertex in
    $\{ c_0, \dots, c_d\}$, contradicting our assumption.

    It is easily seen that the algorithm can be implemented in
    polynomial time, completing the proof. 
\end{proof}

We next consider the case of four paths starting from only two sources.

\begin{lemma}\label{lem:two-source-4flow}
    For $k=4, s=2, c=2$, and $t \in \{2, 3, 4\}$, the $(k, s, t, c)$-DP problem can be
  solved in polynomial time.  
\end{lemma}
\begin{proof}
 Let $(G, x_1, x_2, y_1,y_2,y_3, y_4, \sigma, \tau)$ be given. Again, we first prove the statement for $t=4$ and the cases $t=3$ and $t=2$ will follow directly, as the algorithm still works if pairs of vertices from $\{y_1,y_2,y_3, y_4\}$ are not distinct.
  If more than two of the vertices $y_1,y_2,y_3,y_4$ are equal, the input instance has no solution and can be rejected.

  Without loss of generality~we assume that $\sigma(1) = \sigma(2)
  = x_1$ and $\sigma(3) = \sigma(4) = x_2$ and $\tau(i)=y_i$, for all $1\leq i\leq 4$. 
    We proceed with a description of an algorithm that iteratively and alternately constructs sequences $\bar c_i$ and $\bar d_i$ respectively such that either $V(\bar c_{i-1}) \subset V(\bar{c_i})$ is a strict subset (and equivalently for~$\bar{d_i}$) or a decision to the instance is found.  We start with sequences~$\bar c_0=x_1$ and~$\bar d_0=x_2$ (which are not necessarily cut sequences at this point). The iterative definition of our cut sequences is now as follows; the first two iterations mark the base case.
  \begin{itemize}
      \item[$\bar c_1:$]  By construction, if $P_1', \dots, P_4'$ is a solution to the
  instance such that $P_i'$ starts in $\sigma(i)$ and ends in $\tau(i)$, then~$x_2\in V(P_3')\cap V(P_4')$, and thus~$P_1',P_2'$ cannot use it. Whence, if there is no path in $G^c_1 \coloneqq G - V(\bar d_0)$ from $x_1$ to $\tau(1)$ or from $x_1$
     to $\tau(2)$, then the input instance has no solution and the algorithm
   stops here. Otherwise we set $\bar c_1$ as the sequence of cut
   vertices from $x_1$ to $\{ y_1, y_2 \}$ in~$G^c_1$. Thus, in general, if there is a solution to the problem, then~$V(\bar c_1) \subseteq V(P_1') \cap V(P_2')$. Also, recall that~$V(P^1_{G^c_1}(\bar c_1)) \cap V( P^2_{G^c_1}(\bar c_1 )) = V(\bar c_1)$, and by construction~$P^1_{G^c_1}(\bar c_1) \cap V( \bar d_0) = \emptyset$ as well as~$P^2_{G^c_1}(\bar c_1) \cap V( \bar d_0) = \emptyset$.
   
   Note further that if there are two
  internally disjoint paths from $x_1$ to $\{ y_1, y_2 \}$ in~$G^c_1$, then either $\bar
  c_1 = x_1$ if~$y_1 \neq y_2$ or~$\bar c_1 = x_1,y_1$ otherwise. (The former is the only case where the cut sequence does not increase, for it was no cut sequence to start with).
  
     \item[$\bar d_1:$]  By construction, if $P_1', \dots, P_4'$ is a solution to the input
  instance such that $P_i'$ starts in $\sigma(i)$ and ends in $\tau(i)$, then~$V(\bar c_1) \subseteq V(P_1') \cap V(P_2')$ as seen in the previous step.
  Therefore, if in $G^d_1 \coloneqq G - V(\bar c_1)$ there is no path from $x_2$ to
  $\tau(3)$ or from $x_2$ to $\tau(4)$, the input instance has no
  solution and the algorithm stops here. If there exist two internally disjoint paths~$P_3,P_4$ from~$x_2$ to~$\{y_3,y_4\}$ in~$G^d_1$ then there is a solution to the instance formed by these paths together with~$P^1_{G^c_1}(\bar c_1)$ and~$P^2_{G^c_1}(\bar c_1)$---neither of these paths uses~$x_2$---for no vertex is used by three paths.
  Otherwise, if no two such paths exist, there is a sequence $\bar d_1$ of cut vertices in $G^d_1$ from $x_2$ to $\{ y_3, y_4 \}$ such that~$V(\bar{d_0}) \subset V(\bar{d_1})$ is a strict subset. Note that the paths~$P_3',P_4'$ of the solution satisfy~$V(\bar{d_1}) \subseteq V(P_3') \cap V(P_4')$. Also, recall that~$V(P^1_{G^d_1}(\bar d_1)) \cap V( P^2_{G^d_1}(\bar d_1 )) = V(\bar d_1)$, and by construction~$P^1_{G^d_1}(\bar d_1) \cap V( \bar c_1) = \emptyset$ as well as~$P^2_{G^d_1}(\bar d_1) \cap V( \bar c_1) = \emptyset$.
  \end{itemize}

Suppose now that the sequences~$\bar{c_0},\ldots,\bar c_{i}$ and~$\bar{d_0},\ldots,\bar d_{i}$ have inductively been constructed for some~$ i \geq 1$ maintaining the following two invariants.

    \begin{itemize}
        \item[$(\star,c)$] $V(\bar c_1) \subset \ldots \subset V(\bar c_i)$ where~$\bar c_i$ is a cut sequence in~$G^c_i\coloneqq G-V(\bar d_{i-1})$. It holds~$V(P^1_{G^c_i}(\bar c_i)) \cap V(P^2_{G^c_i}(\bar c_i)) = V(\bar c_i)$ and~$V(P^1_{G^c_i}(\bar c_i)) \cap V(\bar d_{i-1})= \emptyset$ as well as~$V(P^2_{G^c_i}(\bar c_i)) \cap V(\bar d_{i-1})= \emptyset$. Further, if $P_1', \ldots, P_4'$ is a solution to the input instance such that $P_i'$ starts in $\sigma(i)$ and ends in $\tau(i)$, then~$V(\bar c_i) \subseteq V(P_1') \cap V(P_2')$.
        \item [$(\star,d)$] $V(\bar d_1) \subset \ldots \subset V(\bar d_i)$, where~$\bar d_i$ is a cut sequence in~$G^d_i\coloneqq G- V(\bar c_i)$. It holds~$V(P^1_{G^d_i}(\bar d_i)) \cap V(P^2_{G^d_i}(\bar d_i)) = V(\bar d_i)$ and~$V(P^1_{G^d_i}(\bar d_i)) \cap V(\bar c_{i})= \emptyset$ as well as~$V(P^2_{G^d_i}(\bar d_i)) \cap V(\bar c_i)= \emptyset$. Further, if $P_1', \ldots, P_4'$ is a solution to the input instance such that $P_i'$ starts at $\sigma(i)$ and ends in $\tau(i)$, then~$V(\bar d_i) \subseteq V(P_3')\cap V(P_4')$.
    \end{itemize}

As discussed above,~$\bar c_1$ and~$\bar d_1$ satisfy~$(\star,c)$ and~$(\star,d)$ respectively. Our algorithm continues inductively as follows.

\begin{itemize}
    \item[$\bar c_{i+1}:$]  By~$(\star,d)$, if $P_1', \dots, P_4'$ is a solution to the input
  instance such that $P_i'$ starts in $\sigma(i)$ and ends in $\tau(i)$, then~$V(\bar d_i)\subseteq V(P_3') \cap V(P_4')$, and thus~$P_1',P_2'$ cannot use it. Whence, if there is no path in $G^c_{i+1}\coloneqq G - V(\bar d_i)$ from $x_1$ to $\tau(1)$ or from $x_1$
     to $\tau(2)$---in particular if~$V(G)\setminus V(\bar d_i) = \emptyset$---then the input instance has no solution and the algorithm stops here. Otherwise, we set $\bar c_{i+1}$ as the sequence of cut vertices from $x_1$ to $\{ y_1, y_2 \}$ in~$G^c_{i+1}$. 
   
   If~$\bar c_{i+1} = \bar c_{i}$, then this means that~$P^j_{G^c_{i+1}}(\bar c_{i+1}) = P^j_{G^c_{i}}(\bar c_i)$ for~$j = 1,2$. But since the former are defined on~$G-V(\bar d_i)$ this means that~$P^1_{G^c_i}(\bar c_i)$ and~$P^2_{G^c_i}(\bar c_i)$ do not visit~$V(\bar d_i)$ and further, by~$(\star,c)$, they satisfy~$V(P^1_{G^c_i}(\bar c_i)) \cap V(P^2_{G^c_i}(\bar c_i)) = V(\bar c_i)$. But by~$(\star,d)$ the paths~$P^1_{G^d_i}(\bar d_i)$ and~$P^2_{G^d_i}(\bar d_i)$ are disjoint from~$V(\bar c_i)$  and satisfy $V(P^1_{G^d_i}(\bar d_i) \cap V(P^2_{G^d_i}(\bar d_i)) = V(\bar d_i)$. Altogether then~$\{P^1_{G^c_{i+1}}(\bar c_{i+1}),P^2_{G^c_{i+1}}(\bar c_{i+1}),P^1_{G^d_i}(\bar d_i),P^2_{G^d_i}(\bar d_i)\}$ is a solution to the instance and the algorithm stops here. 
   
   Thus, we may assume that~$V(\bar{c_i}) \subset V(\bar c_{i+1})$ is a strict subset; by construction using the standard arguments as above~$(\star,c)$ is still valid.

   \item[$\bar d_{i+1}:$] The argument for~$\bar d_{i+1}$ is analogous to the argument for~$\bar c_{i+1}$, ending either with the algorithm returning an outcome to the instance or a cut sequence~$\bar d_{i+1}$ satisfying~$(\star,d)$.
\end{itemize}

Finally, since there are only finitely many vertices in~$G$ and in each iteration the algorithm either returns or increases the size of~$V(\bar c_i)$ and~$V(\bar d_i)$, we inevitably end up with~$G - V(\bar c_i)$ or~$G- V(\bar d_i)$ being empty, in which case the algorithm also returns correctly as discussed above. It is easily seen that each of the iterations can be implemented in polynomial time.

\end{proof}

By symmetry, we immediately get the next corollary.

\begin{corollary}\label{cor:two-source-4flow}
  \begin{enumerate}  
  \item For $k=3, t=2, c=2$, and $s \in \{2, 3\}$, the $(k, s, t, c)$-DP problem can be
  solved in polynomial time.  
\item     For $k=4, t=2, c=2$, and $s \in \{2, 3, 4\}$, the $(k, s, t, c)$-DP problem can be
  solved in polynomial time.  
  \end{enumerate}
\end{corollary}

\section{Directed Tree-Width}
\label{sec:dtw}

In this section, we recall the definition of directed tree-width and
give an algorithm solving the \textsc{$k$-Half-Integral Directed Paths} problem on
instances of bounded directed tree-width. See
\cite{KreutzerO2014,KreutzerK2018} for surveys of directed width parameters.

An \emph{arborescence} is a directed graph~$T$ containing exactly one source vertex~$r \in V(T)$---that is, the in-degree of~$r$ is~$0$---such that the underlying undirected graph of~$T$ is a tree; we call~$r$ the \emph{root of~$T$}. Note that the root of an arborescence is uniquely defined by~$T$. We denote by $\preceq_T$ the \emph{ancestor relation}, where $u \preceq_T v$ if the (unique) directed path from the root $r$ of~$T$ to $v \in V(T)$ contains $u \in V(T)$. By $T_v$ we denote the subtree of $T$ with root $v$ and vertex set $V(T_v) = \{ u \in V(T) \sth v \preceq_T u \}$.

\begin{definition}[Directed Tree-Width]
  A \emph{directed tree-decomposition} of a digraph $G$ is a triple
  $\TTT := (T, \beta, \gamma)$, where $T$ is an arborescence, $\beta \sth V(T) \rightarrow \PPP(V(G))$ and
  $\gamma \sth E(T) \rightarrow \PPP(V(G))$ are functions such that
  \begin{enumerate}
  \item for every $v \in V(G)$ there is exactly one $t \in V(T)$ with $v \in
    \beta(t)$ and
  \item for every $e = (u, v) \in E(T)$ there is no  directed walk
    in $G - \gamma(e)$ which starts and ends at a vertex in $\beta(T_v) := \bigcup \{ \beta(v') \sth v' \in V(T_v) \}$ and
    contains a vertex of $V(G) \setminus \beta(T_v)$. 
  \end{enumerate}
  We refer to the sets $\beta(t)$ as \emph{bags} and to the sets
$\gamma(e)$ as \emph{guards}.
  
  For $t \in V(T)$ we define $\Gamma(t) := \beta(t) \cup \{ \gamma(e) \sth e \sim_T t \}$, where
  $e \sim_T t$ if and only if $e$ is incident to $t$. The \emph{width} $w(\TTT)$ is
  defined as $\max \{ |\Gamma(t)| \sth t \in V(T) \}$.  The \emph{directed tree-width} of $G$, denoted by
  $\dtw(G)$, is the minimum width of any directed tree-decomposition
  of $G$. 
\end{definition}

Johnson et al. \cite{JohnsonRST2001} (implicitly) describe an algorithm for
computing approximate directed tree-decompositions of a digraph $G$ of
width at most $3 \dtw(G)+1$. This can be turned into an
fpt-algorithm, see \cite{KreutzerO2014,KreutzerK2018,CamposLMS2019}.

\begin{theorem}[\cite{CamposLMS2019}, Theorem 3.7]\label{thm:comp-dtw}
  There is an fpt-algorithm which, given a digraph $G$ and an integer $k \in \N$ as
  input, either computes a directed tree-decomposition of $G$ of width
  at most $3k-2$, or certifies that $\dtw(G) > k$. 
\end{theorem}

The algorithm of the previous theorem actually returns a \emph{nice} directed
tree-decomposition. 

\begin{definition}[Nice Directed Tree-Decomposition]
Let~$G$ be a digraph. A directed tree-decomposition $\TTT := (T, \beta, \gamma)$ of $G$ is \emph{nice}, if
\begin{enumerate}
\item for every edge $e = (u, v) \in
E(T)$ the set $\beta(T_v)$ is a strongly
connected component of $G - \gamma(e)$ and
\item for every vertex $t\in V(T)$ with children $c_1\ldots, c_\ell$ the set $\bigcup_{1\leq i\leq \ell}\beta(c_i)\cap\bigcup_{e\sim_T t}\gamma(e)$ is empty.
\end{enumerate}
\end{definition} 

We will use the niceness of the respective tree-decomposition below.

 \begin{definition}[Elementary cylindrical grid]\label{def:cyl-grid}
  An \emph{elementary cylindrical grid} of order~$k$, for some~$k\geq 1$, is a
  digraph~$G_k$ consisting of~$k$ directed cycles~$W_1, \dots, W_k$
  of length~$2k$,
  pairwise vertex disjoint, together with a set of~$2k$ pairwise
  vertex disjoint paths~$H_1, \dots, H_{2k}$ of length~$k-1$ such that \parsep-10pt
  \begin{itemize}
  \item each path~$H_i$ has exactly one vertex in common with each
    cycle~$W_j$ and has one end in~$V(W_1)$ and the other in $
    V(W_k)$,
  \item the paths~$H_1, \dots, H_{2k}$ appear on each~$W_i$ in this
    order, and
  \item for odd~$i$ the cycles~$W_1, \dots, W_k$ occur on all~$H_i$
    in this order and for even~$i$ they occur in reverse order~$W_k,
    \dots, W_1$.
  \end{itemize}
\end{definition}

See Figure~\ref{fig:grid} for an illustration of~$G_4$.

\begin{definition}\label{rem:grid-in-wall}
  The \emph{elementary cylindrical wall}~$\WWW_k$ of order~$k$ is the
  digraph obtained from the elementary cylindrical grid~$G_k$ by replacing every
  vertex~$v$
  of degree~$4$ in~$G_k$ by two new vertices~$v_s, v_t$ connected by
  an edge~$(v_s, v_t)$ such that~$v_s$ has the same in-neighbours as~$v$ and~$v_t$ has the same out-neighbours as~$v$.

  A \emph{cylindrical wall}~$\WWW$ of order~$k$, or \emph{$k$-wall}, is a subdivision of~$\WWW_k$.
\end{definition}

\begin{figure}
  \begin{center}
    \includegraphics[scale=1]{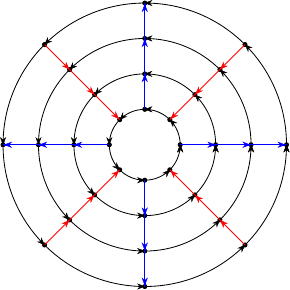}
	\hspace*{1cm}\includegraphics[scale=1]{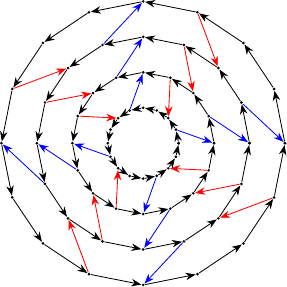}
    \addtolength{\textfloatsep}{-200pt}
    \caption{Cylindrical grid~$G_4$ on the left-hand and a cylindrical wall of order~$4$ on the right-hand side. }\label{fig:grid}\label{fig:wall}
  \end{center}
\end{figure}

In this paper we will only work with cylindrical walls. We will therefore often just say wall to denote a cylindrical wall.
By subwall we denote a subgraph of a wall, which also is a wall.
Cylindrical walls are a natural obstruction to directed tree-width as shown in the next theorem.

\begin{theorem}[\cite{KawarabayashiK2015,HatzelKMM2024a}]\label{thm:directed-grid}
  There is a function $f_{\ref{thm:directed-grid}} \sth \N \rightarrow \N$ such that for all $k \geq 0$ every digraph of
  directed tree-width at least $f(k)$ contains a cylindrical wall of order $k$.
\end{theorem}

We also need the following theorem which captures the cases in which a cylindrical wall can be used for routing a half-integral linkage. 

\begin{theorem}[\cite{GiannopoulouKKK2022}]\label{thm:routing-wall}
  Let $G$ be a digraph, $k\ge 3$ be an integer, and $S := \{s_1, \dots, s_k\}$, $T := \{t_1, \dots, t_k\} \subseteq V(G)$ be sets of size $k$ 
  such that $G$ contains a cylindrical wall $W$ of order $m=k(6k^2+2k+3)$.
  Then in time $\OOO(n^c)$, for some constant $c$ independent of $G, S$, and $T$, one can output either
    \begin{enumerate}[nosep]
  \item a separation $(A_1,B_1)$ of order less than 
    $k$ such that $A_1$ contains $S$ and $B_1$ contains a subwall of $W$ of order at least $m-2k$, 
  \item a separation $(A_2,B_2)$ of order less than $k$ 
  such that $B_2$ contains $T$ and $A_2$ contains a subwall of $W$ of order at least $m-2k$, or    
  \item a set of paths $P_1, \dots, P_k$ in $G$ such that
  $P_i$ links $s_i$ to $t_i$ for $i\in [k]$, and
  each vertex in $G$ is used by at most two of
  these paths.
  \end{enumerate}
\end{theorem}

  We close the section by adapting the disjoint paths algorithm for
  bounded directed tree-width by Johnson et al. \cite{JohnsonRST2001}
  (see also \cite{KreutzerO2014}) to our setting. To this extent let~$k,c,b \geq 1$.

  \problemdef{$k$-Disjoint Directed Paths with Congestion $c$ and Bound $b$}
  { A multi digraph $G$, a set $U \subseteq V(G)$
    of \emph{big vertices},  and distinct vertices $s_1, \dots, s_k,
    t_1, \dots, t_k \in V(G) \setminus U$. 
    Furthermore, for every $u \in U$ a set $F(u)$ of \emph{feasible routings} $\{(e_1, f_1), \dots,
    (e_s, f_s)\}$, where $s \leq b$, such that $e_1, \dots, e_s$ are (not
    necessarily distinct) edges with head $u$ and $f_1, \dots, f_s$
    are (not necessarily distinct) edges with tail $u$.}
  {Decide whether there exist~$(s_i,t_i)$-paths $P_i$
    in $G$ for every~$1 \leq i \leq k$ such that every vertex in $V(G) \setminus U$ is contained in at most $c$ of these paths. Furthermore, for  
    $u \in U$ let
    $R(u) := \{ (e, f) \sth e = (v, u), f = (u, w) \in E(G)$ and there is $1
    \leq i \leq k$ such that $P_i$ contains the subpath $(v, u, w) \}$.
    If $R(u) \not= \emptyset$ then $R(u) \in F(u)$.}

    \begin{theorem}\label{thm:bounded-dtw}
    There is an algorithm running in time $|G|^{O(b\cdot(k+\dtw(G)))}$ that solves
    the $k$-Disjoint Directed Paths Problem with Congestion $2$ and Bound $b$. 
  \end{theorem}

  To prove the theorem we first need some preparation. We fix a digraph $G$, a set $U\subseteq V(G)$ of big vertices, a set $F(u)$ of feasible routings for each~$u \in U$, and terminals~$s_1,\ldots,s_k,t_1,\ldots,t_k \in V(G) \setminus U$ for the exposition in the rest of this section.

  \paragraph{Boundaries.} Let $G$ be a digraph and let $S \subseteq
  V(G)$. Let $P$ be a path in $G$ whose endpoints are not in $S$.
  We define $\lsplit_G(P, S)$ as the set of triples $(e, P', f)$ such
  that $P'$ is a maximal subpath of $P$ entirely  contained in $G[S]$, $e$ is
  the edge of $P$ whose head is the first vertex of $P'$, and $f$ is
  the edge of $P$ whose tail is the last vertex of $P'$.
  For a set $\LLL$ of paths with no endpoint in $S$ we define
  $\lsplit_G(\LLL, S) := \bigcup \{ \lsplit_G(P, S) \sth P \in \LLL \}$.
     Note that if $\LLL$ is only half-integral, then 
 there can be paths $P \not= P' \in \LLL$ such that $\lsplit_G(P, S) = \lsplit_G(P', S)$. We need to keep track of such duplicates and therefore agree that $\lsplit_G(\LLL, S)$ is a multiset.

  The \emph{paths} of $\lsplit_G(\LLL, S)$ are the paths $P'$ in $G$ such
  that $(e, P', f) \in \lsplit_G(\LLL, S)$, for some edges $e, f \in
  E(G)$. In particular~$P'$ is a subpath of some path in~$\LLL$. For simplicity we will abuse notation and write~$\{P_1,\ldots,P_\ell\} = \lsplit_G(\LLL,S)$ to mean the set of paths of~$\lsplit_G(\LLL, S)$ whenever it may not cause confusion. We define the \emph{boundary}  of $\lsplit_G(\LLL, S)$ as the
  multiset $\lBound(\LLL, S) := \{(e, f) \sth$
  there is $P' \subseteq G$ with $(e,  P', f) \in \lsplit_G(\LLL, S) \}$. 
  
  \paragraph{Admissible linkages.} Let $\LLL$ be a set of paths in
  $G$. $\LLL$ is \emph{admissible} if every vertex of $V(G) \setminus U$ appears
  in at most $2$ paths in $\LLL$ and for every $u \in U$ if the set $\LLL(u) := \{ (e, f)
  \sth e = (w, u), f = (u, v)$ and $(w,u,v)$ is a subpath of some path in
  $\LLL$ $\}$ is non-empty then $\LLL(u)$ is a feasible routing in $F(u)$.
  
  Observe that if $\LLL$ is admissible, then every vertex $v \in V(G)$
  occurs in at most $\ell \leq b$ paths contained in $\lsplit_G(\LLL, S)$
  and these must come from sets $\lsplit_G(P_1, S), \dots, \lsplit_G(P_\ell, S)$
  such that $P_i \not= P_j$ whenever $i \not= j$. Here and below we always assume that $b \geq 2$, i.e.,~that big vertices can be used at least as often as the vertices in $V(G) \setminus U$.
  
\paragraph{$w$-Boundedness.} For $w \geq 0$ we say that a set $S \subseteq V(G)$ is \emph{$w$-bounded} if
  for every admissible linkage
  $\LLL$ of order at most $k$ whose endpoints are not in $S$, the set
  $\lsplit_G(\LLL, S)$ contains at most $k+bw$ paths.

  We define $\sBound(S) := \{ \lBound(\LLL, S) \sth \LLL$ is an admissible linkage 
  with no endpoint in $S \}$.
  Note that if $S$ is $w$-bounded then each $B \in \sBound(S)$ contains at
  most $k+bw$ pairs of edges and thus 
  $\sBound(S)$ can contain at most ${ |E(G)| \choose 2(k+bw) }$ sets.

\medskip
  The following is an adaptation of standard results on directed
  tree-decompositions (see e.g.,\cite[Lemma 1.12.8,
  1.12.9]{KreutzerO2014}) that factors in the presence of big vertices.
  
  Let $\TTT := (T, \beta, \gamma)$ be a directed tree-decomposition of $G$ of
  width $w \in \N$. As explained above, the directed tree-decomposition computed by
  the algorithm in \cref{thm:comp-dtw} is nice. This implies that we can assume that the
  children $c_1, \dots, c_\ell$ of a node $t \in V(T)$ are ordered such
  that if $\ell \geq i > j \geq 1$ then $G$ has no edge from $\beta(T_{c_i})$ to
  $\beta(T_{c_j})$. 

  \begin{lemma}\label{lem:split-bounded}
  For every $t \in V(T)$ the set $\beta(T_t)$ is $w$-bounded. Furthermore,
  if $t \in V(T)$ and $s_1, \dots, s_\ell$ are the children of $t$ ordered
  as above then for all $1 \leq j \leq \ell$, $\bigcup\{ \beta(T_{s_i}) \sth i \leq  j\}$ is
  $w$-bounded. 
\end{lemma}
\begin{proof}
  Let $t \in V(T)$. We show that $\beta(T_t)$ is $w$-bounded. This is clear
  if $t$ is the root for then the split is empty. Thus we may assume that $t$ is not the root of
  $T$. Let $e = (p, t) \in E(G)$ be the incoming edge of $t$. Suppose
  $\beta(T_t)$ is not $w$-bounded and let $\LLL$ be an admissible linkage
  witnessing that $|\lsplit_G(\LLL, \beta(T_t))| > k + bw$. Consider a path $L
  \in \LLL$ and let $\{ P_1, \dots, P_r \} = \lsplit_G(L, \beta(T_t))$. Then $L$
  must contain at least $r - 1$ vertices from $\gamma(e)$. For, if $P_1,
  \dots, P_r$ appear on $L$ in this order, then for $1 \leq i < r$, the
  subpath $I_i$ of $L$ starting at the end of $P_i$ and ending at the start
  of $P_{i+1}$ forms a path from $\beta(T_t)$ to $\beta(T_t)$ containing a
  vertex from $G - \beta(T_t)$. Thus $I_i$ must contain a vertex of $\gamma(e)$.
  As every vertex of $\gamma(e)$ can only be used at most $b$ times, this implies that
  $|\lsplit_G(\LLL, \beta(T_t))| \leq k+bw$; a contradiction.

  We prove the second part by induction on $1 \leq i \leq \ell$. For $i = 1$ this follows from Part 1. Now suppose that $X := \bigcup \{ \beta(T_{s_j}) \sth 1 \leq j < i\}$ is $w$-bounded and let $Y := \beta(T_{s_i})$. Recall that there are no edges from $X$  to $Y$ in $G$. Suppose that $X \cup Y$ was not $w$-bounded. Let $\LLL$ be an admissible linkage witnessing this, i.e.,~$|\lBound(\LLL, X \cup Y)| > k+bw$. Consider a path $L \in \LLL$ and its sequence $L_1, \dots, L_p$ of subpaths which are paths of $\lsplit_G(L, X \cup Y)$, ordered in the order in which they appear on $L$. For $1 \leq j < p$ let $I_j$ be the subpath of $L$ starting at the end of $L_j$ and ending at the start of $L_{j+1}$. Then each $I_j$ must contain a vertex of $\Gamma(t)$. For, if $I_j$ contains a vertex in $V(G) \setminus \beta(T_t)$, then it must also contain a vertex of $\gamma(e)$, where $e$ is the edge of $T$ with head $t$ ($t$ cannot be the root in this case). But if $I_j$ does not contain a vertex of $V(G) \setminus \beta(T_t)$, then its internal vertices must be from $\beta(t) \cup \bigcup \{ \beta(T_{s_r}) \sth i < r \leq \ell \}$. But there are no edges from $X \cup Y$ to $\bigcup \{ \beta(T_{s_r}) \sth i < r \leq \ell \}$. Thus, $I_j$ must contain a vertex of $\beta(t)$.

  As every vertex of $\Gamma(t)$ can only be used at most $b$ times, the claim follows. 
\end{proof}

The next lemma is the technical core of our algorithm.

\begin{lemma}\label{lem:comp-w-bounded}
  Let $A, B \subseteq V(G)$ be disjoint such that $A$, $B$, and $A \cup B$ are $w$-bounded. Then
  $\sBound(A \cup B)$ can be computed in polynomial time from the sets
  $\sBound(A)$ and $\sBound(B)$. 
\end{lemma}
\begin{proof}
  To compute $\sBound(A \cup B)$ we need to decide for every possible multiset $X$
  of at most $k+bw$ pairs of edges whether there is an admissible
  linkage $\LLL$ such that $\lBound(\LLL, A \cup B) = X$.
  As $A$ and $A \cup B$ are $w$-bounded, for  every admissible linkage $\LLL$ we have
  $|\lBound(\LLL, A)|, |\lBound(\LLL, A \cup B)| \leq k+bw$. Furthermore,
  $|\lBound(\LLL, B)| \leq k+bw$ as $B$ is $w$-bounded. Hence, to compute $\lBound(\LLL, A \cup B)$ it suffices to consider
  linkages in $G[A], G[A \cup B],$ and $G[B]$ with at most $k+bw$ paths. 
  
  Let $X$ be a multiset of at most $k+bw$ pairs $(e, f)$ of edges such that
  $e$ has its head in $A \cup B$ and its tail outside of $A \cup B$ and $f$
  has its tail in $A \cup B$ and its head outside of $A \cup B$.

  Let $Y_A, Y_B$ be multisets of pairs of edges such that
  \begin{itemize}
  \item if $(e, f) \in Y_A$ then $e \in (V(G) \setminus A) \times A$ and $f \in A \times (V(G) \setminus A)$, and
  \item if $(e, f) \in Y_B$ then $e \in (V(G) \setminus B) \times B$ and $f \in B \times (V(G) \setminus B)$.
  \end{itemize}
  For any multiset $U$ of pairs of edges we define $E(U)$ as the multiset $\{ e, f \sth (e,
  f) \in U \}$. 
  An edge $e$ in $E(Y_A) \cup E(Y_B)$ is called a
  \emph{cross edge} if $e$ has one end in $A$ and the other in $B$.
  An edge of $E(Y_A) \cup E(Y_B)$
  that is not a cross edge is a \emph{boundary edge}. 

  Let $Y_A = \{ (e_1^A,
    f_1^A), \dots, (e_m^A, f_m^A) \}$ and let $Y_B = \{ (e_1^B,
    f_1^B), \ldots, (e_n^A, e_n^B) \}$. Let $C_A$ be the multiset of cross edges in $E(Y_A)$ and $C_B$ be
    the multiset of cross edges in $E(Y_B)$.

    We call $Y_A$ and $Y_B$ \emph{compatible} if the following
    conditions are met.
    \begin{itemize}
    \item There is a bijection $\chi$ between $C_A$ and $C_B$ such that
      for each $f_i^A \in C_A$ there is $1 \leq j \leq n$ such that
      $\chi(f_i^A) = e_j^B$ and $f_i^A = e_j^B$. Likewise, for each $f_i^B \in C_B$ there is $1 \leq j \leq m$ such that $\chi(f_i^B) = e_j^A$ and $f_i^B = e_j^A$.
    \item Let $e \in \{ e_1^A, \dots, e_m^A, e_1^B, \dots, e_n^B \}$ be an
      edge that is not a cross edge. The \emph{path of $e$ in
        $Y_A, Y_B$} is the maximal sequence
      $((e_1, f_1), \dots, (e_p, f_p))$ such that $e_1 = e$, $f_p$ is
      not a cross edge, and for all $1<j\leq p$, $e_j$ is a cross edge and
      $e_j = \chi(f_{j-1})$. Furthermore, if $e = e_i^A$, for some
      $1 \leq i \leq m$, then for all $1 \leq q \leq p$, if $q$ is odd then the
      pair $(e_{q}, f_{q}) \in Y_A$, and if $q$ is even then the pair
      $(e_{q}, f_{q}) \in Y_B$. Otherwise, if $e = e_i^B$, for some
      $1 \leq i \leq n$, then for all $1 \leq q \leq p$, if $q$ is odd then 
      $(e_{q}, f_{q}) \in Y_B$, and if $q$ is even then 
      $(e_{q}, f_{q}) \in Y_A$.

      Then for every $e \in \{ e_1^A, \dots, e_m^A, e_1^B, \dots, e_n^B \}$
      that is not a cross edge, the path of
      $e$ in $Y_A, Y_B$ exists. Let $f(e) \in \{ f_1^A, \dots,
      f_m^A, f_1^B, \dots, f_n^B \}$ be the last edge on the path. Then
      for every $f \in \{ f_1^A, \dots,
      f_m^A, f_1^B, \dots, f_n^B \}$ that is not a cross edge there is
      an $e = e(f) \in \{ e_1^A, \dots, e_m^A, e_1^B, \dots, e_n^B \}$
      that is not a cross edge such that $f = f(e)$.
    \end{itemize}
    Note here that~$\chi$ is a bijection between multisets. Thus the map from $\{e_1^A,\ldots,e_n^A,e_1^B,\ldots,e_m^B\}$ to $\{f_1^A,\ldots,f_n^A,f_1^B,\ldots,f_m^B\}$ which maps~$e$ to~$f(e)$ is a bijection between multisets. 
    
    If $Y_A, Y_B$ are compatible we call the set $\{ (e, f(e)) \sth e  \in \{
    e_1^A, \dots, e_m^A, e_1^B, \dots, e_n^B \} \text{ is not a cross edge} \}$ the \emph{boundary} of the pair $Y_A, Y_B$ and denote it by
    $\lBound(Y_A, Y_B)$.  Let $X = \{ (e_1, f_1), \ldots, (e_\ell, f_\ell) \}$. We call $(X, Y_A, Y_B)$
    \emph{compatible} if $X = \lBound(Y_A, Y_B)$.
  \begin{Claim}\label{lem:comp-w-bounded:1}
    \begin{enumerate}
    \item Let $\LLL$ be an admissible linkage with ends outside of
      $A \cup B$. Then
      $(\lBound(\LLL, A \cup B), \lBound(\LLL, A), \lBound(\LLL, B))$ is
      compatible.
      \item Conversely, let $(X, Y_A, Y_B)$ be compatible and
        let $\LLL_A$ and $\LLL_B$ be admissible linkages with all internal
        vertices of $\LLL_A$ in $A$ and all internal vertices of $\LLL_B$ in
        $B$ such that
        $Y_A = \{ (e, f) \sth $ there is a path in $\LLL_A$ starting with $e$
        and ending with $f \}$ and $Y_B  = \{ (e, f) \sth $ there is a path in $\LLL_B$ starting with $e$
        and ending with $f \}$.

        Then the digraph obtained from the union of $\LLL_A$ and $\LLL_B$
        contains an admissible linkage $\LLL$ such that $\lBound(\LLL, A
        \cup B) = X$.
    \end{enumerate}
  \end{Claim}
  \begin{ClaimProof}
  \begin{enumerate}
      \item Let $\LLL$ be an admissible linkage with
    ends outside of $A \cup B$. Let $X = \lBound(\LLL, A \cup B)$, $Y_A =
    \lBound(\LLL, A)$, and $Y_B = \lBound(\LLL, B)$.
    Then $X$ contains the pairs $(e, f)$ of edges with exactly one end
    outside of $A \cup B$ such that there is a subpath $L'$ of a path $L \in
    \LLL$ that starts in $e$ and ends in $f$.
    By symmetry we may assume that $e$ has its head in $A$. Let $e_1 =
    e$ and let $e_2, \dots, e_\ell$ be the cross edges appearing on $L'$
    in this order. Let $e_{\ell + 1} = f$. Then for all $1 \leq i \leq \ell$ if
    $i$ is odd then 
    $(e_{i}, e_{i+1}) \in Y_A$ and if $i$ is even then $(e_{i}, e_{i+1}) \in Y_B$.
    We define $\chi_L(e_{i+1}) =
    e_{i}$, for all $1 \leq i \leq \ell-1$ with $i$ odd.

    Let $\chi$ be the union of $\chi_L$ taken over all paths in $\LLL$. Then
    the path of $e$ in $Y_A, Y_B$ with respect to $\chi$ is defined
    and starts in $e$ and ends in $f$.
    Thus $\chi$ witnesses that
    $(X, Y_A, Y_B)$ is compatible.
      \item Let $(X, Y_A, Y_B)$ be compatible witnessed by
    the bijection $\chi$. Consider $(e, f) \in X$. By definition of
    compatible triples, $(e, f) \in \lBound(Y_A, Y_B)$ and the path
    $P_{e,f}$ of $e$ in $Y_A, Y_B$ exists and ends in $f$. Let
    $P_{e,f} = \big((e_1, f_1), (e_2, f_2), \ldots, (e_\ell, f_{\ell})\big)$, with $e = e_1$ and $f = f_\ell$. Recall that $e_i = f_{i-1}$ for all $1 < i \leq \ell$. By 
    symmetry we assume that the head of $e$ is in $A$. Then,
    for all $1 \leq i \leq \ell$, if $i$ is odd then $(e_i, f_{i}) \in Y_A$ and
    if $i$ is even then  $(e_{i}, f_{i}) \in Y_B$. For each $1 \leq i \leq \ell$ with
    $i$ odd let $L_i$ be the path in $\LLL_A$ starting at $e_i$ and
    ending at $f_i = e_{i+1}$. Similarly, for all $1 \leq i \leq \ell$ with $i$ even
    let $L_i$ be the path in $\LLL_B$ starting at $e_i$ and
    ending at $f_i = e_{i+1}$; these exist by assumption. Then $L_{e,f} := L_1 + L_2 +
    \dots + L_\ell$ is a path in $G$ which starts at $e$, ends at $f$, and
    has all internal vertices in $A \cup B$.
    We claim that  $\LLL = \{ L_{e,f} \sth (e, f) \in X\}$ is the required
    admissible linkage. What is left to show is that no non-big vertex---vertices in~$V(G) \setminus U$---is used more than twice and that for every big vertex the feasibility
    condition is satisfied. But this follows immediately from the
    condition that $\LLL_A$ and $\LLL_B$ are admissible and that if $u \in A \cap U$
    is a big vertex in $A$ then if some path in $\LLL$ contains a subpath
    $(w, u, v)$ then this is already a subpath of some path in $\LLL_A$.
    Thus the feasibility condition on $u$ is satisfied as it is
    satisfied in $\LLL_A$. The analogous argument holds for $u \in B \cap U$.
    Thus $\LLL$ is an admissible linkage as required.
  \end{enumerate}
  This concludes the proof of the claim.
  \end{ClaimProof}
 To compute $\sBound(A \cup B)$ from $\sBound(A)$ and $\sBound(B)$ we
 proceed as follows.
 Let $X := \{ \lBound(Y_A, Y_B) \sth Y_A \in \sBound(A), Y_B \in \sBound(B)$
 and $Y_A$ and $Y_B$ are compatible$\}$. \cref{lem:comp-w-bounded:1} implies that  $X = \sBound(A \cup B)$.
 It is easily seen that $X$ can be computed in polynomial time given
 $\sBound(A)$ and $\sBound(B)$.
\end{proof}

We are now ready to prove \cref{thm:bounded-dtw}.

\smallskip

\begin{proof}[Proof of \cref{thm:bounded-dtw}]
  We first apply \cref{thm:comp-dtw} to compute a nice directed tree-decomposition $(T, \beta, \gamma)$ of $G$ of width $w \leq 3\cdot\dtw(G)$. Recall that we assume the instance $(G, s_1, \ldots, s_k, t_1, \ldots t_k)$ to be in normal form. This implies that we can always ensure that $\beta(r) = \{ s_1, \ldots, s_k, t_1, \ldots, t_k \}$ for the root $r$ of $T$.  
  
  By induction on the directed tree-decomposition $(T, \beta, \gamma)$ we
  compute for each $t \in V(T)$ other than the root $r$ of $T$ the set
  $\sBound(\beta(T_t))$. Clearly, $G$ contains an admissible linkage
  linking $s_i$ to $t_i$, for $1 \leq i \leq k$, if and only if
  $\{ ( e_i, f_i) \sth 1 \leq i \leq k \} \in \sBound(\beta(T_{t'}))$, where
  $t'$ is the unique successor of the root $r$ of $T$, $e_i$ is the
  unique outgoing edge of $s_i$, and $f_i$ is the unique incoming edge
  of $t_i$. For convenience we write $\sBound(s)$ for
  $\sBound(\beta(T_s))$, for $s \in V(T)$.

  If $t$ is a leaf of $T$, then $\beta(t)$ has size bounded by $w$ and
  thus we can compute $\sBound(t)$ by brute-force. 
  Now suppose $t \in V(T)$ has children $c_1, \dots, c_\ell$,
  $\ell \geq 1$, ordered such that for $i > j$ there is no edge from
  $\beta(T_{c_j})$ to $\beta(T_{c_i})$.  Assume that
  $\sBound(c_i)$ has already been computed for all $i$. If $\ell > 1$ we
  first compute for all $1 \leq j \leq \ell$ the set
  $\sBound(\bigcup \{ \beta(T_{c_i}) \sth 1 \leq i \leq j\}$ by applying
  \cref{lem:comp-w-bounded} repeatedly. 

  We then compute $\sBound(\beta(t))$. As any set $\beta(t)$ of size $\leq w$ is necessarily $w$-bounded, another application
  of \cref{lem:comp-w-bounded} yields the set $\sBound(\beta(T_t))$.
\end{proof}

\section{Uncrossing $2$-Separations}
\label{sec:uncross}

The following observation follows immediately from the \cref{def:separation} of directed separations.
Let $G$ be a digraph and $W = (C_1, \dots, C_k, P_1, \dots, P_{2k})
\subseteq G$ be a cylindrical wall of order $k$. 
If $S = (A, B)$ is a separation of $G$ of order $< k$ then exactly
one of the two sides $X \in \{ A, B\}$ of $S$ contains a cycle $C_i$, for some $1 \leq i
\leq k$. We say that \emph{$X$ contains the majority of $W$}.

\begin{definition}[Quadrants and Crosses]
  Let $S_1 = (A_1, B_1)$ and $S_2 = (A_2, B_2)$ be separations in $G$
  such that $\partial^+(S_i) = A_i$, for $i=1,2$. For $i=1,2$ let $X_i = A_i \cap B_i$ be the separators of $S_1, S_2$, respectively. 

  The \emph{quadrants} of the pair $(S_1, S_2)$ are the sets $Q_T := A_1 \cap
  A_2$, $Q_L := A_1 \cap B_2$, $Q_R := B_1 \cap A_2$, and $Q_B := B_1 \cap
  B_2$. We call $Q_T$ the \emph{top} and $Q_B$ the \emph{bottom quadrant}.
  $Q_L$ and $Q_R$ are called the \emph{middle quadrants}.
  
  The \emph{corners} of $(S_1, S_2)$ are the sets $Q_i \cap (X_1 \cup X_2)$,
  for $i \in \{ T, B, L, R\}$. The corner of $Q_T$ is called the
  \emph{top corner} and the corner of $Q_B$ is the \emph{bottom
    corner}.

  We say that  $(S_1, S_2)$ \emph{crosses} if $Q \setminus (X_1 \cup X_2) \not=
  \emptyset$ for every quadrant $Q$  of $(S_1, S_2)$. Otherwise $S_1$ and
  $S_2$ are uncrossed.
\end{definition}

The next lemma states the well-known submodularity property of directed separations. See e.g.,~\cite{GiannopoulouKKK2022}.

\begin{lemma}[Submodularity]\label{lem:submodularity}
  For $i=1,2$ let $S_i = (A_i, B_i)$ be separations in a digraph $G$
  such that $\partial^+(S_i) = A_i$. Then $X_T = (A_1 \cap A_2, B_1 \cup B_2)$ and $X_B
  = (A_1 \cup A_2, B_1 \cap B_2)$ are separations in $G$ and $|X_T| + |X_B| \leq
  |S_1| + |S_2|$.
\end{lemma}
Using submodularity we get the following.

\begin{lemma}\label{lem:uncross}
  Let $G$ be a digraph, $\bar s, \bar t$ be tuples of vertices of $V(G)$ such that there is no separation $(A, B)$ of order $1$ in $G$ with $\bar s, \bar t \subseteq A$ and $B \setminus A \not= \emptyset$.
  
  For $i=1,2$ let $S_i := (A_i, B_i)$ be $2$-separations in $G$ and let $W_i$
  be cylindrical walls of order at least $4$ such that $\bar s, \bar t \subseteq
  A_1 \cap A_2$ and the majority of $W_i$ is contained in $B_i$.
  Furthermore, at least one of $S_1, S_2$ separates the majority of $W_1$ from the majority of $W_2$.
  
  Then either there is a separation $(A, B)$ of order $2$ with~$\bar{s},\bar{t} \subseteq A$ and $B_1 \setminus A_1 \cup B_2 \setminus A_2 \subseteq B$, or there are separations $S_i' := (A_i', B_i')$ of order $2$, for $i=1,2$, such that for all $1 \leq i \leq 2$, $\bar s, \bar t \subseteq A_i'$,
  $B_i'$ contains the
  majority of $W_i$, and $B_1' \setminus A_1' \cap B_2' \setminus A_2' = \emptyset$. Furthermore, $S_1' = S_1$ or $S_2' = S_2$. 
\end{lemma}
\begin{proof}
  If  $B_1 \subseteq B_2$ or $B_2\subseteq B_1$ we are done, thus assume otherwise. Similarly we may assume that $B_1 \setminus A_1 \cap B_2 \setminus A_2 \not= \emptyset$, for else we are done.
  
  We first consider the case where $\partial^+(S_i) = A_i$, for
  all $1 \leq i \leq 2$. In this
  case $\bar s$ and $\bar t$ 
  are in the top quadrant of the pair $\big((A_1, B_1), (A_2, B_2)\big)$. By \cref{lem:submodularity}, $X_T := (A_1 \cap A_2, B_1 \cup
  B_2)$ and $X_B := (A_1 \cup A_2, B_1 \cap B_2)$ are both separations and
  $|X_T| + |X_B| \leq 4$. 
  
  If $|X_T| \leq 2$ then we are done as $X_T$ satisfies the requirements of the lemma. Otherwise, if $|X_T|
  \geq 3$, then $|X_B| < 2$. But $B_1 \setminus A_1 \cap B_2 \setminus A_2 \not= \emptyset$. 
  Hence $X_B = (A_1 \cup A_2, B_1 \cap B_2)$ is a separation of order at most $1$ with  $\bar s, \bar t \subseteq A_1 \cup A_2$ and $B_1 \cap B_2 \setminus (A_1 \cup A_2) \not= \emptyset$, which contradicts the assumption of the lemma.  This implies that
  $|X_B| \geq 2$ and therefore $X_T$ satisfies the requirements of the
  lemma. The case where $\partial^+(S_i) = B_i$ for $i=1, 2$ is symmetric, with the r\^ole of bottom and top exchanged.

  Now suppose that, for some $i \in \{1, 2\}$, there are no cross edges from $B_i \setminus A_i$ to $A_i \setminus B_i$ and no cross edges from $A_{3-i} \setminus B_{3-i}$ to $B_{3-i} \setminus A_{3-i}$. By symmetry we may assume that $i=1$. Let $X_T := (A_1 \cap B_2, A_2 \cup B_1)$ and $X_B := (A_1 \cup B_2, A_2 \cap B_1)$. By
  construction, $\partial^+(X_T) = A_1 \cap  B_2$ and $\partial^+(X_B) = A_1 \cup B_2$.
  By \cref{lem:submodularity}, $|X_T| + |X_B| \leq 4$. 
  If $(A_1 \cap B_2) \setminus A_2 = \emptyset$, that is, if there is no vertex in the top quadrant of $S_1, S_2$ that is not already contained in the separators, then we can simply take $X_1$ as solution as then $S_1$ satisfies the requirements of the lemma. Similarly, if $(B_1 \cap A_2) \setminus A_1 = \emptyset$ then $S_2$ satisfies the requirement of the lemma. 

  Thus we may now assume that the top and the bottom quadrant of the pair $S_1, S_2$ are both non-empty. As $G$ has no separation of order at most $1$
  separating $\bar  s, \bar t$ from some non-empty part of the graph, we can infer that $|X_T|
  = |X_B| = 2$. 
    Furthermore, as at least one of $S_1$ and $S_2$ separates the majority of $W_1$ from the majority of $W_2$, either the majority of $W_2$ is in the top quadrant or the majority of $W_1$ is in the bottom quadrant. 
    In the first case we set $(A_1', B_1') = (A_1, B_1)$ and $(A_2', B_2') = (B_1 \cup A_2, A_1 \cap B_2)$. In the second case we set $(A_1', B_1') = (A_1 \cup B_2, B_1 \cap A_2)$ and $(A_2', B_2') = (A_2, B_2)$.
\end{proof}

\section{Constructing a Directed Tree-Decomposition Along Two-Cuts}
\label{sec:constructing-dtw}

In this section we prove our main algorithmic result.

\main*

Towards this aim we will construct a directed
tree-decomposition along $2$-separations between the terminals and
large enough walls. We complement the definition of arborescence by that of \emph{inbranchings}; an \emph{inbranching} is a directed graph~$T$ such that~$T$ admits exactly one sink vertex~$r \in V(T)$---that is~$r$ has out-degree~$0$---such that the underlying undirected graph of~$T$ is a tree.

Throughout the remainder of this section, let~$(G, s_1, s_2, s_3, t_1, t_2, t_3)$ be an instance of the
\textsc{$3$-Half-Integral Directed Paths} problem in normal form, and set $\bar s := (s_1, s_2, s_3)$ and~$ \bar t := (t_1, t_3, t_3)$.

\paragraph{Eliminating $1$-separations.} As a first step we delete all vertices in $V$
that are not on a path from some $s_i$ to some $t_j$. Clearly, any such
vertex cannot be part of any solution, so they can simply be removed.
If we remove a terminal at this step then we can stop immediately and
reject the instance. 
From now on we will therefore assume that every vertex in $V(G)$ lies
on a path from some $s_i$ to some $t_j$. 

As a second step we eliminate $1$-separations in~$G$. Suppose there is a
separation $(A, B)$ of $G$ of order $1$ with $\bar s, \bar t \subseteq A$ and
$B \setminus A \not=\emptyset$. Let $v$ be the unique vertex in $A \cap B$.
Then we contract $B$ into the vertex $v$, that is, we replace every
edge $(u, w)$ with $w \in B$ and $u \in A \setminus \{ v \}$ by the edge $(u, v)$
and every edge $(w, u)$ with $w \in B$ and $u \in A \setminus \{ v \}$ by the edge
$(v, u)$. (Note that there can only be at most one type of edge.)
Let $G'$ be the resulting instance. It is easily seen that $(G', \bar
s, \bar t)$ has a solution if, and only if, $(G, \bar s, \bar t)$ has
a solution. Clearly, we can compute $1$-separations as above in polynomial time.
Thus, from now on we assume that there are no such separations of order
$1$. 

\paragraph{$\SSS$-reductions.} For $i \geq 0$ we construct sets $\SSS_i$ of $2$-separations $S = (A, B)$
in $G$ with the following properties.
\begin{enumerate}
\item If  $(A, B) \in \SSS_i$ then  $\bar s, \bar t \subseteq A$ and there is
  a wall $W$ of order at least $\wallbound$ whose majority is
  contained in $B$.
\item For distinct separations $(A, B), (A', B') \in \SSS_i$ we have $B \cap
  B' \subseteq A \cap A'$. 
\end{enumerate}

Let $\SSS$ be a set of separations satisfying the properties above.
We define the \emph{$\SSS$-reduction $G_\SSS$ of $G$} as the multidigraph obtained from
$G$ by contracting for each $S_j := (A_j, B_j) \in \SSS$ the set $B \setminus A$ of vertices
into a new vertex $c_j$. Let~$\ell \coloneqq \Abs{\SSS}$.
Formally, we define a sequence $G_j$ of multidigraphs, for $0 \leq j \leq \ell$, as
follows.
We set $G_0 := G$. Now suppose $G_j$ has already been defined for $j <
|\SSS|$. Let $B := B_{j+1} \setminus A_{j+1}$. Then $G_{j+1}$ is the multidigraph with vertex set $(V(G_j) \setminus
B)\, \dot\cup\, \{ c_{j+1}\}$, where $c_{j+1}$ is a fresh vertex. The
edges of $G_{j+1}$ are the edges of $G_j$ that have no endpoint in $B$ and for all $u \in A_{j+1}$
all edges $(u, c_{j+1})$ such that there is some $v \in B$ with $(u, v)
\in E(G_j)$, and all edges $(c_{j+1}, u)$ such that there is some $v \in
B$ with $(v, u) \in E(G_j)$. 

Finally, we define $G_{\SSS} := G_\ell$.
Note that $G_\SSS$ is a multidigraph as there may be several edges from a
vertex  $v \in A_j$ to $B_j \setminus A_j$ or vice versa which will result in
parallel edges from $v$ to $c_j$. We could easily reduce this to
simple graphs without changing the problem by subdividing each of
these edges once. But this would come at the expense of notational
overhead. 

The next claim establishes the relevant properties of $G_\SSS$ we need in
the remainder of the proof.

\begin{lemma}\label{lem:s-reduction}
  For $1 \leq j \leq \ell$, let $\{ u_j, u_j' \} = A_j \cap B_j$.
  \begin{enumerate}
  \item Either $c_j$ has outgoing edges to $u_j$ and $u_j'$ and no other
    outgoing edges or $c_j$ has incoming edges from $u_j$ and $u'_j$ and
    no other incoming edges.
  \item Let $W$ be a wall of order $\geq 3$ in $G_\SSS$. Then there is a
    wall $W' \subseteq G$ of the same order as $W$ and $W'$ is a subdivision
    of $W$. 
  \end{enumerate}
\end{lemma}
\begin{proof}
  We prove by induction on $j$ that the $\{ S_1, \dots, S_j\}$-reduction
  $G_j$ of $G$ satisfies the conditions of the claim.

  For $j=1$, $S_1 = (A_1, B_1)$ is a separation of order $2$. Let $u, u'$
  be the two vertices in $A_1 \cap B_1$. Without loss of generality we
  assume that there are no cross edges from $B_1 \setminus A_1$ to $A_1 \setminus
  B_1$ (the other case is symmetric). 
  By construction, there is a wall $W_1$ of order at least
  $\wallbound$ whose majority is in $B_1$. Thus $B_1 \setminus A_1$ is not
  empty.
  $G_1$ is obtained from $G$ by deleting all vertices
  in $B_1 \setminus A_1$, adding a new vertex $c_1$ and redirecting all edges
  with exactly one endpoint in $B_1 \setminus A_1$ to $c_1$. 
  As $G$ has no separation $(X, Y)$  of order $1$ with $\bar s,
  \bar t \subseteq X$ and $Y \setminus X \not= \emptyset$, there must be 
  edges from $B_1 \setminus A_1$ to $u$ and also to $u'$ 
  and from every vertex $v \in B_1 \setminus A_1$ there are two
  internally disjoint paths from $v$ to $u$ and $u'$, respectively. 
  Thus,  $G_1$ contains edges $(c_1, u)$ and $(c_1, u')$.
  Furthermore, $c_1$ has no other outgoing edge.
  This proves the first part of the statement.
  
  To prove the second part, let $W$ be a wall of order $\geq 3$ in $G_1$. If
  $c_1 \not\in V(W)$, then $W \subseteq G$ and there is nothing to show.
  Otherwise, $c_1 \in V(W)$. As $W$ has maximum degree $3$, this means
  that either $c_1$ only has $1$ outgoing edge in $W$ or only has one
  incoming edge in $W$. Suppose first that $c_1$ only has one incoming
  edge $(v, c_1)$ in $W$. Then it either has one or two outgoing edges
  with endpoints in $\{u, u'\}$. By construction, there is $w \in B_1
  \setminus A_1$ with $(v, w) \in E(G)$. As noted above, there are two
  internally disjoint paths $P_u, P_{u'}$ in $G[B_1]$ starting at $w$ and ending at
  $u$ and $u'$, respectively. Thus, if we replace the edge $(v, c_1)$
  in $W$ by $(v, w)$ and the edges $(c_1, u)$ and $(c_1, u')$ by $P_u$
  and $P_{u'}$ (or by only one of the paths if $c_1$ has outdegree $1$
  in $W$) then this yields a subdivision of $W$ in $G$.

  A similar argument can be used in the case where $c_1$ has indegree two in
  $W$. We give more details of this case in the induction step. 
  This completes the base step for $j=1$.

  Now let $j>1$ and let $S_j = (A_j, B_j)$. By construction, $G_j$ is
  obtained from $G_{j-1}$ by deleting all vertices in $B_j \setminus A_j$ and
  adding a fresh vertex $c_j$ with all edges as defined above. 
  Again we assume without loss of generality that in $G$ there are no
  cross edges from $B_j \setminus A_j$ to $A_j \setminus B_j$. By construction, in
  $G_{j-1}$ a vertex $u$ in $B_j \setminus A_j$ only has an edge to some vertex
  $c_i$,  $i < j$, if before there was an edge $(u, w)$ to a vertex
  in $w \in B_i \setminus A_i$. By the second condition of $\SSS$ above,
  $B_j \cap (B_i \setminus A_i) = \emptyset$. As by assumption there are no edges from
  $B_j \setminus A_j$ to $A_j \setminus B_j$, $v$ cannot have an edge to
  $c_i$ in $G_{j-1}$. It follows that in $G_j$ the fresh vertex $c_j$
  has outdegree $2$ with the vertices in $B_j \cap A_j$ being its only
  outneighbours.

  Now let $W$ be a wall of order $\geq 3$ in $G_j$. Again, if $W$ does
  not contain $c_j$ then $W \subseteq G_{j-1}$ and, by induction hypotheses,
  there is a subdivision $W' \subseteq G$ of $W$ in $G$.
  Otherwise we argue as in the case $j=1$. That is, $c_j$ can have
  at most three incident edges in $W$, at least one incoming and at least
  one outcoming edge. Without loss of generality we now assume that
  $c_j$ has two incoming edges $(w_1, c_j), (w_2, c_j)$ and that $u \in
  B_j \cap A_j$ is its sole outneighbour.
  Then in $G_{j-1}$ there are vertices $v_1, v_2 \in B_j \setminus A_j$ with
  $(w_1, v_1), (w_2, v_2) \in E(G_{j-1})$. 
  Let $P_1$ be a path in
  $G_{j-1}[B_j \setminus A_j]$ from $v_1$ to $v_2$ which is disjoint from $u$.
  Let $P_2$ be a path from $v_2$ to $u$. Then $P_1 \cup P_2$ contains an
  inbranching with root $u$ that contains $v_1$ and $v_2$. That is,
  there is a vertex $w \in B_j \setminus A_j$ and internally disjoint
  paths $I_1, I_2, O$ in $G_{j-1}[B_j \setminus A_j]$ such that $I_1$ and
  $I_2$ link $v_1$ and $v_2$ to $w$ and $O$ links $w$ to $u$. We now
  replace in $W$ the vertex $c_j$ by $w$, the edge $(w_1, c_j)$ by the concatenation of
  $(w_1,v_1)$ and $I_1$, the edge $(w_2, c_j)$ by the concatenation of $(w_2,v_2)$ and $I_2$, and the edge $(c_j, u)$ by $O$
  to get a subdivision $W'$ of $W$ in $G_{j-1}$. By induction hypothesis
  there is a subdivision $W''$ of $W'$ in $G$, which proves the second
  part of the statement. 
\end{proof}

\paragraph{Reducing along~$2$-separations.}  We will now reduce the input instance to an equivalent instance of bounded directed tree-width. We start the construction by setting $\SSS_0 := \emptyset$.
Now, given $\SSS_i$, let $G_i$ be the $\SSS_i$-reduction of $G$.  If $\dtw(G_i)
\leq \dtwbound$, then the process stops.
Otherwise, by~\cref{thm:directed-grid}, there is a wall $W$ of order
at least $\wallbound$ in $G_i$. By \cref{lem:s-reduction}, $G$
contains a subdivision $W'$ of $W$. We now apply \cref{thm:routing-wall} to $(G, W', \bar s, \bar t)$.
If this yields a solution, we are done. Otherwise, we obtain a
separation $S_{i+1} = (A, B)$ of order $2$ with $\bar s \cup \bar t \subseteq A \setminus B$ and the
majority of $W'$ in $B$.  

Let $\SSS := \SSS_i \cup \{S_{i+1}\}$. $\SSS$ satisfies the first condition we
require of the set $\SSS_{i+1}$, but it may not yet satisfy Condition 2. 
Inductively we construct sets $\SSS^j$ as follows. We will maintain the invariant that there is a separation $S \in \SSS^j$ such that  $\SSS^j \setminus \{ S \}$ does not contain a pair of separations that are crossed. Furthermore, $S$ is  the only separation which separates the majority of $W'$ from $\bar{s} \cup \bar{t}$. We start by setting $\SSS^1 := \SSS$. Now
suppose $\SSS^j$ has already been constructed. If $\SSS^j$ contains $S_1 = (A_1, B_1)$
and $S_2 = (A_2, B_2)$ with $B_1 \cap B_2 \not\subseteq A_1 \cap A_2$, then we proceed as
follows. By the invariant above, $W'$ is separated from $\bar{s} \cup \bar{t}$ by exactly one of $S_1$ and $S_2$, say by $S_1$. Thus we can apply \cref{lem:uncross} to $S_1$ and $S_2$. If the outcome
is a single separation $S = (A, B)$ then we set $\SSS^{j+1} := \SSS^j \setminus \{
S_1, S_2 \} \cup \{ S \}$. In this case, \cref{lem:uncross} guarantees that
$B_1 \setminus A_1 \cup B_2 \setminus A_2 \subseteq B$  and thus the majority of $W'$ and of $W_2$ are in $B$, where $W_2$ is the
wall whose majority is contained in $B_2$.  If
\cref{lem:uncross} yields two separations $S_1', S_2'$, then we set
$\SSS^{j+1} := \SSS^j \setminus \{ S_1, S_2 \} \cup \{ S_1', S_2' \}$. The lemma
guarantees that $S'_1$ and $S'_2$ are uncrossed and that $S_i'$ separates $\bar s, \bar t$ from the  majority of the wall $W_i$, for $i=1,2$, where we set $W_1 = W'$. Furthermore, either $S_1 = S_1'$ or $S_2 = S_2'$. In both cases at most one of $S_1', S_2'$ can cross any other separation in  $\SSS^{j+1}$ and therefore we maintain the invariant above.

In each iteration we either reduce the number of separations in $\SSS^j$
or reduce the number of separations that are crossed. Thus
the process stops after a linear number of steps with a set $\SSS^\ell$
of separations. We define $\SSS_{i+1} := \SSS^\ell$. 

\smallskip

Observe that in each step from $\SSS_i$ to $\SSS_{i+1}$  we find a new wall $W_i$
  in the $\SSS_i$-reduction $G_i$ of $G$ and add a separation to $\SSS_i$ that
  separates $W_i$ from $\bar s \cup \bar t$. Therefore, after at most
  $\ell \leq |V(G)|$ iterations, we arrive at an $\SSS_\ell$-reduction $G_\ell$ of
  directed tree-width at most $\dtwbound$.

  This completes the first step of our algorithm. Let $\SSS := \SSS_\ell = \{
  S_1, \dots, S_r \}$ and let $G_\SSS$ be the $\SSS$-reduction of $G$. 
  
\paragraph{Reducing the instance.} Our next goal is to
  apply \cref{thm:bounded-dtw} to $(G_\SSS, \bar s, \bar t, U)$, where $U$ is the set of
  fresh vertices $c_1, \dots, c_r$ in the $\SSS$-reduction of $G$. 
  To apply \cref{thm:bounded-dtw} we need to compute the sets $F(c_i)$
  of feasible routings. 

  For each $1 \leq i \leq r$ we compute $F(c_i)$ as follows. Let $S_i = (A_i,
  B_i)$ be the associated separation in $\SSS$ and let $D := G[B_i]$.  Suppose $c_i$
  only has $2$ outgoing edges to vertices $u, u'$ where $\{ u, u'\} = A_i \cap B_i$. The case where $c_i$ only has incoming edges from $u, u'$ is symmetric. 
  Let $C \subseteq E(G)$ be the set of edges with tail in $A_i$ and head in $B_i \setminus A_i$ and let $C' \subseteq E(G)$ be the set of edges with head $u$ or $u'$ and tail in $B_i
  \setminus A_i$. We define a function $\lambda$ such that for all $e = (v, w) \in C$ we set $\lambda(e) := (v, c_i)$ and for all $f = (w, v) \in C'$ we set $\lambda(e) := (c_i, v)$.
  For every choice of an edge $e \in C$ and $f$ in $C'$ add the pair
  $(\lambda(e), \lambda(f))$ to the multiset $F(c_i)$. Similarly, for every choice of edges $e, e' \in
  C$ and $f, f' \in C'$ add $\{(\lambda(e), \lambda(f)), (\lambda(e'), \lambda(f'))\}$ to $F(c_i)$.

  For every choice of edges 
  $e_j = (v_j, w_j) \in C$ and $f_j = (x_j, u_j) \in C'$, for $1 \leq j \leq 3$,
  we apply \cref{cor:two-source-4flow} as follows. By construction, every edge in $C'$ has either $u$ or $u'$ as head. 
  We now apply Part 1 of \cref{cor:two-source-4flow} to the instance $(D, w_1, w_2, w_3, u, u', \sigma, \tau)$, where $\sigma(i) := w_i$ and $\tau(i) := u_i$, for $1 \leq i \leq 3$ (or to the corresponding instance in case the $w_i$ are not distinct).  
  If the outcome of the lemma is a solution, then we add
  $\{(\lambda(e_1), \lambda(f_1)), (\lambda(e_2), \lambda(f_2)), (\lambda(e_3), \lambda(f_3))\}$ to $F(c_i)$. 
  
  Finally, for every choice of edges 
  $e_j = (v_j, w_j) \in C$ and $f_j = (x_j, u_j) \in C'$, for $1 \leq j \leq 4$,
  we apply Part 2 of \cref{cor:two-source-4flow} to the instance $(D, w_1, \dots, w_4, u, u', \sigma, \tau)$ where $\sigma(i) = w_i$ and $\tau(i) = u_i$, for all $1 \leq i \leq 4$. (Again if some of the $w_i$ are not distinct we use the appropriate instance.) If the outcome is a solution, then we add
  $\{(\lambda(e_1), \lambda(f_1)), (\lambda(e_2), \lambda(f_2)), (\lambda(e_3), \lambda(f_3)), (\lambda(e_4), \lambda(f_4))\}$ to $F(c_i)$.
   
    Observe that a solution can contain at most $4$ (maximal) subpaths starting or ending at $A_i \cap B_i$ with internal vertices in $B_i$. Thus we do not need to consider combinations of more than $4$ edges from $C$ and $C´$, respectively.
  Let $\FFF = \{ F(c_i) \sth 1 \leq i \leq r \}$.

  \begin{lemma}
    The instance $(G_\SSS, \bar s, \bar t, U, \FFF)$ of \textsc{$3$-Disjoint Directed Paths with Congestion $2$ and Bound $4$} has a solution if, and
    only if, the \textsc{$3$-Half-Integral Directed Paths} instance $(G, \bar s, \bar t)$ has a solution. 
  \end{lemma}
  \begin{proof}
    For the forward direction, let $P_1, P_2, P_3$ be a set of paths in $G$ such that $P_i$
    is an $(s_i,t_i)$-path for~$i\in \{1,2,3\}$ and no vertex appears in all three paths. Let
    $S = (A, B) \in \SSS$ be a separation. By assumption, $\bar s \cup \bar t
    \subseteq A$. Thus the paths $P_i$ all start and end in $A$. For $1 \leq i \leq
    3$ let $\PPP$ be the set of maximal subpaths of the paths $P_i$ in $G[B]$ that contain a vertex of $B \setminus A$. For each $P \in \PPP$ with $P \subseteq P_i$, for some $1 \leq i \leq 3$, let $e(P)$
    be the edge in $E(P_i)$ whose head is the first vertex of $P$ and
    let $f(P)$ be the edge in $E(P_i)$ whose tail is the last vertex
    of $P$. Then the paths in $\PPP$ witness that the sequence $((\lambda(e(P)),
    \lambda(f(P)))_{P \in \PPP}$ appears in $F(c)$, where $c$ is the fresh vertex associated with the selected separation $S$ (recall the definition of $\lambda$ from above). Thus, $(G_\SSS, \bar s, \bar t, U,
    \FFF)$ has a solution which can be obtained from $P_1, P_2, P_3$ by
    replacing for each $P \in \PPP$ which is a subpath of $P_i$, for some
    $1 \leq i \leq 3$, the subpath $(e(P), P, f(P))$ by the subpath $(u, c,
    v)$, where $u$ is the tail of $e(P)$ and $v$ is the head of
    $t(P)$.

    Towards the other direction, let $P_1, P_2, P_3$ be a solution of $(G_\SSS, \bar s, \bar t, U,
    \FFF)$. For each $1 \leq i \leq \ell$ let $F_i$ be the set of triples $(u,
    c_i, v)$ such that for some $1 \leq j \leq 3$ the path $P_j$ contains the subpath $(u, c_i, v)$.
    By definition,  $\{ ((u, c_i), (c_i, v)) \sth (u, c_i, v) \in F_i \} \in
    F(c_i)$ for all $1 \leq i \leq \ell$. This implies that for each triple $(u, c_i, v) \in F_i$ there are
    vertices $u', v' \in B_i \setminus A_i$ and a path $P_{u,v} \subseteq G[B_i \setminus A_i]$
    linking $u'$  to $v'$ such that $(u, u'), (v', v) \in E(G)$ and the
    set $L_i := \{ P_{u,v} \sth (u, c_i, v) \in F_i \}$  is a half-integral linkage.
    Thus for all $1 \leq i \leq \ell$, all $1 \leq j \leq 3$, and all $(u, c_i, v)
    \in F_i$, if $P_j$ contains the subpath $(u, c_i, v)$ then we can
    substitute the subpath by $P_{u,v} \in L_i$.
    Let $P_1', P_2', P_3'$ be the resulting paths. Then  $P_i'$ links
    $s_i$ to $t_i$ and $L = \{ P_1', P_2', P_3' \}$ is a solution in $G$. 
  \end{proof}

  As the last step of our algorithm we apply \cref{thm:bounded-dtw}
  to the instance $(G_\SSS, \bar s, \bar t, U, \FFF)$. By the previous lemma,
  if the algorithm returns a solution then this yields a solution to
  $(G, \bar s, \bar t)$. Otherwise $(G, \bar s, \bar t)$ does not have
  a solution. This concludes the proof of \cref{thm:main}.

\section{NP-Hardness}
\label{sec:hardness}
In this section, we proof \cref{thm:nphard}, which we restate here for convenience.

\nphardness*

The proof generalizes the NP-hardness proof for \textsc{2 Disjoint Paths} in directed graphs by Fortune et al.~\cite{FHW80}.
The main difference is a more involved \emph{switch gadget}.
The gadget is depicted in \cref{fig:switch} and its main properties are as follows.
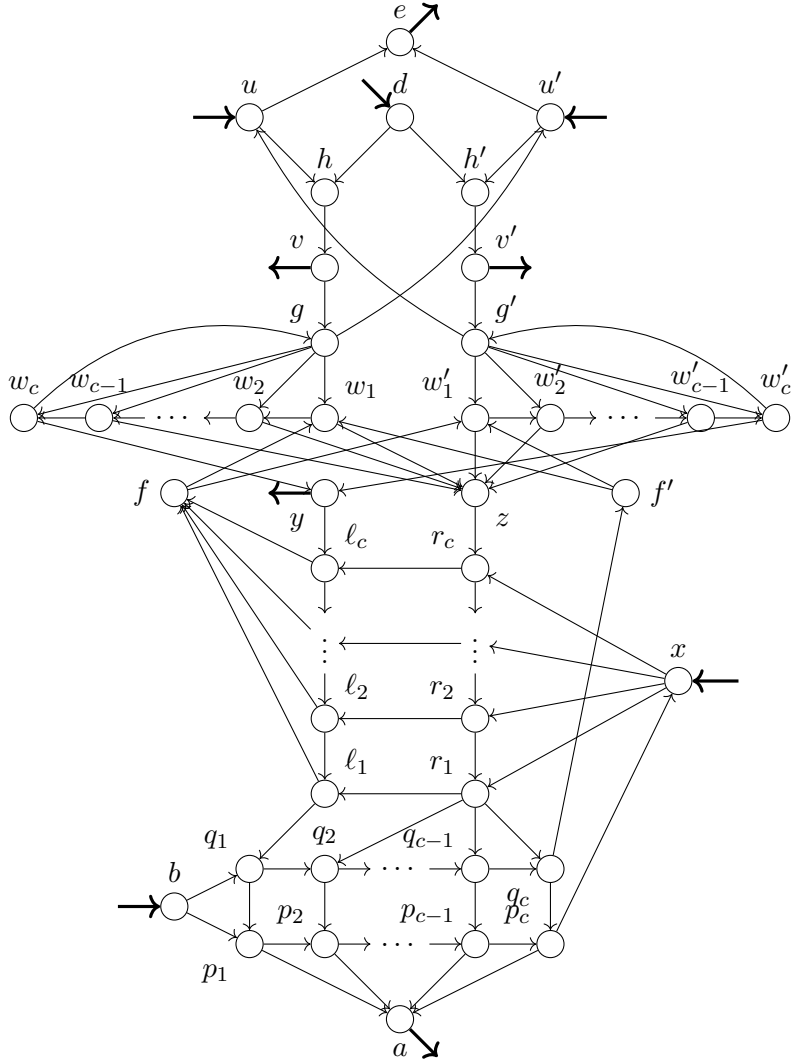
\begin{figure}[t!]
    \centering
    \begin{tikzpicture}
        \node[circle,draw,label=below:$a$] at (0,-2) (a) {};
        \node[circle,draw,label=$x$] at(3.7,2.5) (x) {};
        \node[circle,draw,label=$b$] at (-3,-.5) (b) {};
        \node[circle,draw,label=left:$f$] at(-3,5) (e) {};
        \node[circle,draw,label=right:$f'$] at(3,5) (f) {};
        \node[circle,draw,label=below left:$y$] at(-1,5) (y) {};
        \node[circle,draw,label=below left:$p_1$] at (-2,-1) (p1) {};
        \node[circle,draw,label=above left:$p_2$] at (-1,-1) (p2) {};
        \node at (0,-1) (p3) {$\dots$};
        \node[circle,draw,label=above left:$p_{c-1}$] at (1,-1) (p4) {};
        \node[circle,draw,label=above left:$p_{c}$] at (2,-1) (p5) {};
        \node[circle,draw,label=above left:$q_1$] at (-2,0) (q1) {};
        \node[circle,draw,label=above:$q_2$] at (-1,0) (q2) {};
        \node at (0,0) (q3) {$\dots$};
        \node[circle,draw,label=above left:$q_{c-1}$] at (1,0) (q4) {};
        \node[circle,draw,label=below left:$q_{c}$] at (2,0) (q5) {};
        \node[circle,draw,label=above right:$\ell_1$] at (-1,1) (r1) {};
        \node[circle,draw,label=above right:$\ell_2$] at (-1,2) (r2) {};
        \node at (-1,3) (r4) {$\vdots$};
        \node[circle,draw,label=above right:$\ell_{c}$] at (-1,4) (r5) {};
        \node[circle,draw,label=above left:$r_1$] at (1,1) (s1) {};
        \node[circle,draw,label=above left:$r_2$] at (1,2) (s2) {};
        \node at (1,3) (s4) {$\vdots$};
        \node[circle,draw,label=above left:$r_{c}$] at (1,4) (s5) {};
        \node[circle,draw,label=below right:$z$] at(1,5) (yy) {};
        \foreach \i in {1,2,4,5}{
            \draw[->] (p\i) to (a);
            \draw[->] (q\i) to (p\i);
            \draw[->] (x) to (s\i);
            \draw[->] (r\i) to (e);
            \draw[->] (s\i) to (r\i);
        }
        \draw[->] (r1) to (q1);
        \foreach \i in {2,4,5}{
            \draw[->] (s1) to (q\i);
        }
        \draw[->] (b) to (p1);
        \draw[->] (p1) to (p2);
        \draw[->] (p2) to (p3);
        \draw[->] (p3) to (p4);
        \draw[->] (p4) to (p5);
        \draw[->] (b) to (q1);
        \draw[->] (q1) to (q2);
        \draw[->] (q2) to (q3);
        \draw[->] (q3) to (q4);
        \draw[->] (q4) to (q5);
        \draw[->] (p5) to (x);
        \draw[->] (r5) to (r4);
        \draw[->] (r4) to (r2);
        \draw[->] (r2) to (r1);
        \draw[->] (s5) to (s4);
        \draw[->] (s4) to (s2);
        \draw[->] (s2) to (s1);
        \draw[->] (y) to (r5);
        \draw[->] (yy) to (s5);
        \draw[->] (q5) to (f);
        \node[circle,draw,label=above right:$w_1$] at(-1,6) (w1) {};
        \node[circle,draw,label=above left:$w'_1$] at(1,6) (ww1) {};
        \foreach \i/\n in {2/2,4/c-1,5/c}{
            \node[circle,draw,label=$w_{\n}$] at (-\i,6) (w\i) {};
            \node[circle,draw,label=$w'_{\n}$] at (\i,6) (ww\i) {};
        }
        \node at(-3,6) (w3) {$\dots$};
        \node at(3,6) (ww3) {$\dots$};
        \foreach \i in {1,2,4}{
            \draw[->] (w\i) to (yy);
            \draw[->] (ww\i) to (yy);
        }
        \draw[->] (w1) to (w2);
        \draw[->] (w2) to (w3);
        \draw[->] (w3) to (w4);
        \draw[->] (w4) to (w5);
        \draw[->] (ww1) to (ww2);
        \draw[->] (ww2) to (ww3);
        \draw[->] (ww3) to (ww4);
        \draw[->] (ww4) to (ww5);
        \draw[->] (e) to (w1);
        \draw[->] (f) to (w1);
        \draw[->] (e) to (ww1);
        \draw[->] (f) to (ww1);
        \draw[->] (w5) to (y);
        \draw[->] (ww5) to (y);
        \node[circle,draw,label=$d$] at(0,10) (c) {};
        \node[circle,draw,label=$e$] at(0,11) (d) {};
        \node[circle,draw,label=$u$] at(-2,10) (l10) {};
        \node[circle,draw,label=$u'$] at(2,10) (t10) {};
        \node[circle,draw,label=$h$] at(-1,9) (l9) {};
        \node[circle,draw,label=$h'$] at(1,9) (t9) {};
        \node[circle,draw,label=above left:$v$] at(-1,8) (l8) {};
        \node[circle,draw,label=above right:$v'$] at(1,8) (t8) {};
        \node[circle,draw,label=above left:$g$] at(-1,7) (l7) {};
        \node[circle,draw,label=above right:$g'$] at(1,7) (t7) {};
        \foreach \j in {1,2,4,5}{
            \draw[->] (l7) to (w\j);
            \draw[->] (t7) to (ww\j);
        }
        \draw[->,bend left=30] (w5) to (l7);
        \draw[->,bend right=30] (ww5) to (t7);
        \foreach \d in {l,t}{
            \draw[->] (\d8) to (\d7);
            \draw[->] (\d9) to (\d8);
            \draw[->] (\d10) to (\d9);
            \draw[->] (c) to (\d9);
            \draw[->] (\d10) to (d);
        }
        \draw[->,bend right=15] (l7) to (t10);
        \draw[->,bend left=15] (t7) to (l10);
        \draw[->,very thick] (a) to (.5,-2.5);
        \draw[->,very thick] (-3.75,-.5) to (b);
        \draw[->,very thick] (-.5,10.5) to (c);
        \draw[->,very thick] (d) to (.5,11.5);
        \draw[->,very thick] (4.5,2.5) to (x);
        \draw[->,very thick] (y) to (-1.75,5);
        \draw[->,very thick] (-2.75,10) to (l10);
        \draw[->,very thick] (2.75,10) to (t10);
        \draw[->,very thick] (l8) to (-1.75,8);
        \draw[->,very thick] (t8) to (1.75,8);
    \end{tikzpicture}
    \caption{A switch. The thick arcs enter or leave the gadget.}
    \label{fig:switch}
\end{figure}
\begin{lemma}
    \label{lem:switch}
    Let~$c \geq 2$ and consider the switch gadget in \cref{fig:switch}.
    If each vertex is part of at most~$c$ paths, $c$ paths leave the gadget at vertex~$a$, $c$ different paths enter it at vertex~$b$, yet another~$c-1$ paths enter the gadget at vertex~$x$, and none of these paths start or end within the gadget, then exactly~$c-1$ of the above paths leave the gadget at vertex~$y$, exactly~$c$ of them leave at vertex~$e$, and the $c$ paths leaving at~$a$ enter the gadget at vertex~$d$.
    Moreover, in any such routing exactly~$c$ paths pass through the vertices~$x$ and~$y$ and depending on the routing of the paths leaving at~$a$, exactly~$c$ paths pass through~$g$ and~$h$ or through~$g'$ and~$h'$.
\end{lemma}
\begin{proof}
    We assume that each vertex is part of at most~$c$ paths and no path start or ends in the gadget.
    Then, we say that each vertex has a \emph{capacity} of~$c$ and each path passing through a vertex uses one unit of that capacity.
    We call a vertex that is already part of~$c$ paths \emph{saturated}.
    Note that no further path can go through a saturated vertex.

    For the sake of notational convenience, we will bundle all paths entering the gadget at~$b$ or~$x$ and call these paths the blue paths.
    We say that all paths leaving at~$a$ are red.
    It will be convenient to think about the reverse of red paths.
    We therefore say that red paths start in~$a$ and red paths use arcs in the opposite direction as shown in \cref{fig:switch}.
    We will show that if $c$ blue paths enter at~$b$, $c-1$ blue paths enter in~$x$, and~$c$ red paths enter in~$a$, then all red paths leave at~$d$, $c-1$ blue paths leave the gadget at~$y$ and~$c$ blue paths leave it at~$e$.
    Moreover, in any such routing vertices~$x$ and~$y$ are saturated and so are (i) vertices~$g$ and~$h$ or~(ii) vertices~$g'$ and~$h'$.
    To do so, we start with the assumed paths at their respective first vertices ($a,b,$ and~$x$) and over time show that certain paths have to continue in certain ways.

    First, consider the~$c$ blue paths entering at~$b$.
    If all these paths reach~$q_c$ via the paths~$(b,q_1,q_2,\ldots,q_c)$, then note that no red path can leave the gadget.
    Hence, at least one such blue path reaches the vertex~$p_c$ and then it has to go to vertex~$x$ as vertex~$a$ is saturated by the~$c$ red paths (recall that all red and blue paths are distinct).
    Since~$c-1$ other blue paths start in~$x$, this proves that~$x$ is saturated in any feasible solution and exactly one blue path entering at vertex~$b$ does not start with the path~$(b,q_1,q_2,\ldots,q_c,f')$.
    Hence, $c-1$ blue paths start with this path and in order for the~$c$ red paths to leave the gadget, exactly one has to start with the paths~$(a,p_i,q_i)$\footnote{Technically, the paths do not need to start with the claimed path but they have to contain all three vertices in the relative order. We assume without loss of generality that all solutions are minimal in the sense that if a path contains two vertices~$u$ and~$v$ and~$(u,v)$ is an arc (and the path visits~$u$ before~$v$), then it takes the arc directly and not any detour path. Note that replacing a detour by a direct arc only decreases the amount of units of capacity used by each vertex.} for each~$i \in [c]$.
    Moreover, the~$c-1$ red paths \emph{not} going through~$p_1$ and~$q_1$ have to continue with the path~$(r_1,r_2,\ldots,r_c,z)$ as~$x$ and~$q_1$ are saturated.

    This implies that the~$c$ blue paths currently in~$x$ ($c-1$ started there and one coming from~$p_c$) each have to continue with a path~$(x,r_i,\ell_i,f)$ for some distinct~$i \in [c]$.
    Now, the vertex~$r_i$ is saturated for each~$i \in [c]$ and the red path starting with~$(a,p_1,q_1)$ has to continue with the path~$(q_1,\ell_1,\ell_2,\ldots,\ell_c,y)$.
    We will call this red path~$R_1$ from now on.

    Note that~$R_1$ currently ends in~$y$, the other~$c-1$ red paths currently end in~$z$, and all~$2c-1$ blue paths currently end in either~$f$ or~$f'$.
    Since~$f$ and~$f'$ only have arcs towards~$w_1$ and~$w'_1$, exactly~$c$ blue paths pass through one and~$c-1$ paths pass through the other.
    As both cases are completely symmetric, we will only discuss the case where~$c$ paths pass through~$w_1$.
    Since each vertex~$w_i$ and~$w'_i$ with~$i \in [c-1]$ only has an outgoing arc towards~$z$ and~$w_{i+1}$ or~$w'_{i+1}$, respectively, and since vertex~$z$ only has an outgoing arc to~$r_c$ which is fully saturated, all blue paths continue with the subpath~$(w_1,w_2,\ldots,w_c)$ or~$(w'_1,w'_2,\ldots,w'_c)$.
    Hence, all vertices~$w_1,w_2,\ldots,w_c$ are saturated and~$R_1$ has to continue with the subpath~$(y,w'_c,g')$.
    Moreover, the~$c-1$ other red paths have to each continue with a distinct subpath~$(z,w'_i,g')$ for some~$i \in [c-1]$.
    Hence, $g'$ is saturated by the $c$ red paths and all~$c-1$ blue paths currently ending in~$w'_c$ have to continue with the subpath~$(w'_c,y)$.
    Since~$R_1$ also passed through~$y$, this shows that~$y$ is saturated.
    Moreover, as all vertices~$w_1,w'_1,w_2,w'_2,\ldots,w_c,w'_c$ are saturated, and without these vertices, the only way for paths passing through~$y$ to leave the gadget is at~$y$ or~$a$ (and~$a$ is also saturated), all~$c-1$ paths passing through~$w'_c$ and~$y$ leave the gadget at~$y$.
    For the red paths currently at~$g'$, note that as~$g'$ only has incoming arcs from~$v'$ and~$w'_c$ and~$w'_c$ is saturated by~$c-1$ blue paths and one red path, all red paths have to continue with the subpath~$(g',v',h')$.

    Next, the~$c$ blue paths currently ending in~$w_c$ cannot continue with~$y$ as~$y$ is already saturated.
    Hence, all of these paths continue with the paths~$(w_c,g,u')$.
    As~$h'$ is saturated by the~$c$ red paths, the~$c$ blue paths passing through~$u'$ continue to the vertex~$e$ where they leave the gadget.
    Since~$u'$ is saturated by~$c$ blue paths, the~$c$ red paths continue with the vertex~$d$ and have to leave the gadget there.
    Note that we have shown that vertices~$g',h',x,$ and~$y$ are saturated.
    If~$c$ paths pass through~$w'_1$, then the vertices~$g,h,x$ and~$y$ are saturated.
    In either case, all red paths leave the gadget at vertex~$d$, $c-1$ blue path leave the gadget at vertex~$y$ and the remaining~$c$ blue paths leave the gadget at vertex~$e$.
    This concludes the proof.
\end{proof}

With the above lemma at hand, the proof of \cref{thm:nphard} is now basically the same as in the proof by Fortune et al.~\cite{FHW80}.
We present it here for the sake of completeness and slightly change the presentation of the proof to increase readability.

\smallskip

\begin{proof}[Proof of \cref{thm:nphard}]
    First, observe that containment in NP is simple.
    The input consists of at least~$n+k$ bits and a solution consists of~$k$ paths, each consisting of at most~$n$ vertices.
    Hence, a solution can be guessed in $\OOO(nk\log(n))$ non-deterministic time, which is polynomial in the input size.
    Afterwards, it is trivial to check that guessed structure is a directed path from~$s_i$ to~$t_i$ and each vertex appears in at most~$c$ paths.

    The case for~$c=1$ is equivalent to \textsc{Disjoint Paths} for~$k \geq 2$ paths, which was shown to be NP-hard by Fortune et al.~\cite{FHW80}.
    We therefore assume that~$c \geq 2$ and show NP-hardness for~$k=3c-1$.
    Note that by adding additional vertices~$s_i$ and~$t_i$ and an arc~$(s_i,t_i)$ for any~$3c \leq i \leq k$ to the construction, this also shows NP-hardness for any~$k > 3c-1$.
    
    To show NP-hardness for any~$c \geq 2$ and~$k=3c-1$, we fix an arbitrary~$c \geq 2$ and present a reduction from \textsc{3-Sat}.
    To this end, let a formula~$\phi$ in 3-CNF with variables~$\mathcal{V}=\{x_1,x_2,\ldots,x_n\}$ and clauses~$\mathcal{C}=\{C_1,C_2,\ldots,C_m\}$ be given.
    We start with~$2k$ vertices~$s_1,s_2,\ldots,s_k$ and~$t_1,t_2,\ldots,t_k$.
    These will form the terminal pairs~$(s_i,t_i)$ for each~$i \in [k]$.
    Next, we add~$3m$ switch gadgets~$S_1,S_2,\ldots,S_{3m}$.
    We will only refer to the vertices~$a,b,d,e,u,u',v,v',x$ and~$y$ in these gadgets.
    For the sake of conciseness, we will refer to the vertex~$a$ in gadget~$S_i$ as~$a_i$ and the same for all other mentioned vertices.
    Finally, we add vertices~$j_1,j_2,\ldots,j_{n+1}$ and~$o_1,o_2,\ldots,o_{m+1}$ and connect the different parts as follows.
    We add the arc~$(a_1,t_i)$ for each~$i \in [c]$.
    Moreover, we add the arc~$(s_i,b_1)$ for each~$c+1 \leq i \leq 2c$ and the arc~$(s_i,x_1)$ for each~$2c+1 \leq i \leq 3c-1$.
    Next, for each~$i \in [3m-1]$, we connect the gadgets~$S_i$ and~$S_{i+1}$ by adding the arcs~$(a_{i+1},d_i),(e_i,b_{i+1})$, and~$(y_i,x_{i+1})$.
    We add the arcs~$(y_{3m},t_i)$ and~$(e_{3m},t_i)$ for each~$c+1 \leq i \leq 3c-1$.
    We also add the arcs~$(j_{n+1},o_1)$, $(o_{m+1},d_{3m})$, and~$(s_i,j_1)$ for each~$i \in [c]$.
    To conclude the construction, we add variable and clause gadgets.
    To this end, we enumerate the positions in~$\phi$, that is, for a clause~$C_i \in \mathcal{C}$, we say positions~$3(i-1)+1,3(i-1)+2,$ and~$3(i-1)+3$ correspond to it.
    For a variable~$x_i \in \mathcal{V}$, let~$\alpha_i$ and~$\beta_i$ be the number of times~$x_i$ appears positive and negated in~$\phi$, respectively. Additionally, let~$p_1 < p_2 < \ldots < p_{\alpha_i}$ and~$n_1 < n_2 < \ldots < n_{\beta_i}$ be the positions where~$x_i$ appears positive and negated in~$\phi$.
    We construct the following variable gadget for~$x_i$.
    First, we add the arcs~$(j_i,u_{p_1}), (j_i,u_{n_1}), (v_{p_{\alpha_i}},j_{i+1})$, and~$(v_{n_{\beta_i}},j_{i+1})$ (and assume without loss of generality that $x_i$ appears at least once positive and at least once negated).
    Then, for each~$i' \in [\alpha_i-1]$, we add the arc~$(v_{p_{i'}},u_{p_{i'+1}})$.
    Analogously, we add the arc~$(v_{n_{i'}},u_{n_{i'+1}})$ for each~$i' \in [\beta_i-1]$.
    See \cref{fig:reduction} (left side) for an example of variable gadgets.
    \begin{figure}[t!]
        \centering
        \begin{tikzpicture}
            \node[circle,draw,inner sep=3pt,label=below:$s_1$] (s1) at(-1,0) {};
            \node[circle,draw,inner sep=3pt,label=below:$s_2$] (s2) at(-.33,0) {};
            \node at(.33,0) {$\dots$};
            \node[circle,draw,inner sep=3pt,label=below:$s_c$] (sc) at(1,0) {};
            \node[circle,draw,inner sep=3pt,label=left:$j_1$] at(0,1) (a1) {};
            \node[circle,draw,inner sep=3pt,label=left:$u_1$] at(-1,3) (u1) {};
            \node[circle,draw,inner sep=3pt,label=left:$v_1$] at(-1,4) (v1) {};
            \node[circle,draw,inner sep=3pt,label=right:$u_4$] at(1,2) (u4) {};
            \node[circle,draw,inner sep=3pt,label=right:$v_4$] at(1,3) (v4) {};
            \node[circle,draw,inner sep=3pt,label=right:$u_7$] at(1,4) (u7) {};
            \node[circle,draw,inner sep=3pt,label=right:$v_7$] at(1,5) (v7) {};
            \node[circle,draw,inner sep=3pt,label=left:$j_2$] at(0,6) (a2) {};
            \node[circle,draw,inner sep=3pt,label=left:$u_5$] at(-1,8) (u5) {};
            \node[circle,draw,inner sep=3pt,label=left:$v_5$] at(-1,9) (v5) {};
            \node[circle,draw,inner sep=3pt,label=right:$u_2$] at(1,7) (u2) {};
            \node[circle,draw,inner sep=3pt,label=right:$v_2$] at(1,8) (v2) {};
            \node[circle,draw,inner sep=3pt,label=right:$u_8$] at(1,9) (u8) {};
            \node[circle,draw,inner sep=3pt,label=right:$v_8$] at(1,10) (v8) {};
            \node[circle,draw,inner sep=3pt,label=left:$j_3$] at(0,11) (a3) {};
            \node[circle,draw,inner sep=3pt,label=left:$u_3$] at(-1,12) (u3) {};
            \node[circle,draw,inner sep=3pt,label=left:$v_3$] at(-1,13) (v3) {};
            \node[circle,draw,inner sep=3pt,label=left:$u_6$] at(-1,14) (u6) {};
            \node[circle,draw,inner sep=3pt,label=left:$v_6$] at(-1,15) (v6) {};
            \node[circle,draw,inner sep=3pt,label=right:$u_9$] at(1,13) (u9) {};
            \node[circle,draw,inner sep=3pt,label=right:$v_9$] at(1,14) (v9) {};
            \node[circle,draw,inner sep=3pt,label=$j_4$] at(0,16) (a4) {};
            \node[circle,draw,inner sep=3pt,label=$o_1$] at(4,16) (b1) {};
            \node[circle,draw,inner sep=3pt,label=left:$u'_1$] at(3,14) (x1) {};
            \node[circle,draw,inner sep=3pt,label=left:$v'_1$] at(3,13) (y1) {};
            \node[circle,draw,inner sep=3pt,label=right:$u'_2$] at(4,14) (x2) {};
            \node[circle,draw,inner sep=3pt,label=right:$v'_2$] at(4,13) (y2) {};
            \node[circle,draw,inner sep=3pt,label=right:$u'_3$] at(5,14) (x3) {};
            \node[circle,draw,inner sep=3pt,label=right:$v'_3$] at(5,13) (y3) {};
            \node[circle,draw,inner sep=3pt,label=right:$o_2$] at(4,11) (b2) {};
            \node[circle,draw,inner sep=3pt,label=left:$u'_4$] at(3,9) (x4) {};
            \node[circle,draw,inner sep=3pt,label=left:$v'_4$] at(3,8) (y4) {};
            \node[circle,draw,inner sep=3pt,label=right:$u'_5$] at(4,9) (x5) {};
            \node[circle,draw,inner sep=3pt,label=right:$v'_5$] at(4,8) (y5) {};
            \node[circle,draw,inner sep=3pt,label=right:$u'_6$] at(5,9) (x6) {};
            \node[circle,draw,inner sep=3pt,label=right:$v'_6$] at(5,8) (y6) {};
            \node[circle,draw,inner sep=3pt,label=right:$o_3$] at(4,6) (b3) {};
            \node[circle,draw,inner sep=3pt,label=left:$u'_7$] at(3,4) (x7) {};
            \node[circle,draw,inner sep=3pt,label=left:$v'_7$] at(3,3) (y7) {};
            \node[circle,draw,inner sep=3pt,label=right:$u'_8$] at(4,4) (x8) {};
            \node[circle,draw,inner sep=3pt,label=right:$v'_8$] at(4,3) (y8) {};
            \node[circle,draw,inner sep=3pt,label=right:$u'_9$] at(5,4) (x9) {};
            \node[circle,draw,inner sep=3pt,label=right:$v'_9$] at(5,3) (y9) {};
            \node[circle,draw,inner sep=3pt,label=right:$o_4$] at(4,1) (b4) {};
            \node[circle,draw,inner sep=3pt,label=below:$d_9$] at(4,0) (d) {};
            \foreach \i/\j in {s1/a1,s2/a1,sc/a1,a4/b1,b4/d,a1/u1,a1/u4,v4/u7,v1/a2,v7/a2,a2/u5,a2/u2,v2/u8,v5/a3,v8/a3,a3/u3,a3/u9,v3/u6,v6/a4,v9/a4,a4/b1,b1/x1,b1/x2,b1/x3,y1/b2,y2/b2,y3/b2,b2/x4,b2/x5,b2/x6,y4/b3,y5/b3,y6/b3,b3/x7,b3/x8,b3/x9,y7/b4,y8/b4,y9/b4}{
                \draw[->] (\i) -- (\j);
            }
            \foreach \i in {1,2,...,9}{
                \draw[color=color\i,->,thick] (u\i) -- (v\i) node[midway,left] {\color{black} \i};
                \draw[color=color\i,->,thick] (x\i) -- (y\i) node[midway,left] {\color{black} \i};
            }
        \end{tikzpicture}
        \caption{An example of the reduction behind \cref{thm:nphard} for the formula \begin{equation*}
            \phi=(x_1 \lor \overline{x_2} \lor x_3) \land (\overline{x_1} \lor x_2 \lor x_3) \land (\overline{x_1} \lor \overline{x_2} \lor \overline{x_3}).
        \end{equation*}
        Switches are not fully shown and only indicated by two arcs of the same color (and with the same label). These arcs represent paths of length three in the switch from either~$u$ to~$v$ or from~$u'$ to~$v'$.} 
        \label{fig:reduction}
    \end{figure}
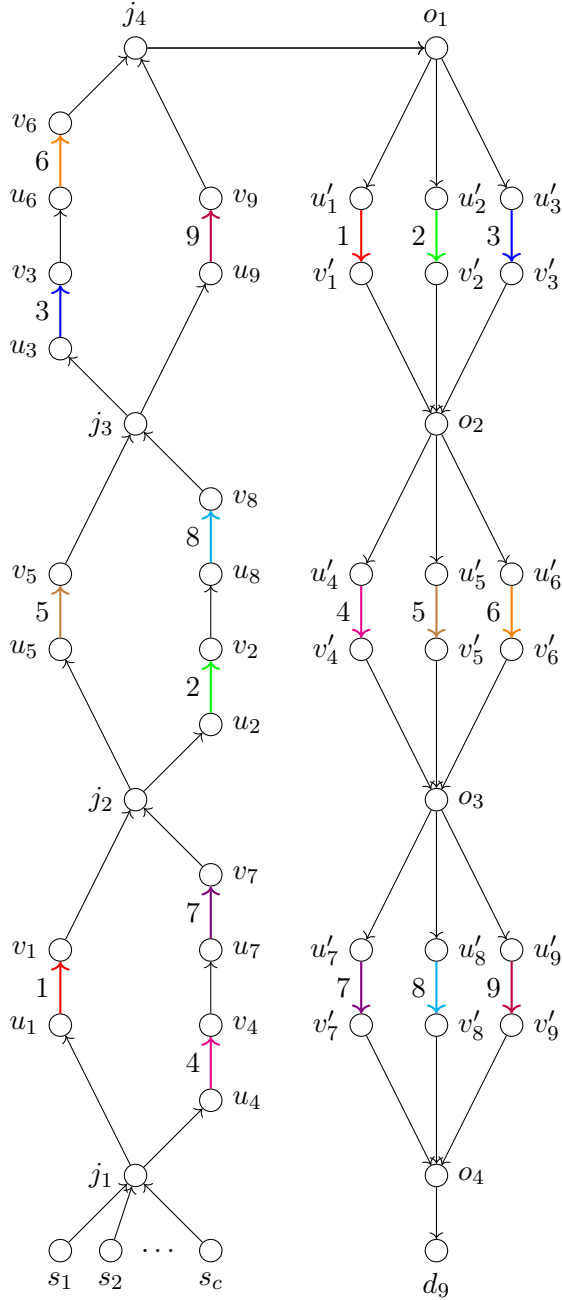
    To conclude the construction, we add the following clause gadget for each clause~${C_i \in \mathcal{C}}$.
    We add the six arcs~$(o_i,u'_{3(i-1)+1}),(o_i,u'_{3(i-1)+2}),(o_i,u'_{3(i-1)+3}),(v'_{3(i-1)+1},o_{i+1}),(v'_{3(i-1)+2},o_{i+1}),$ and~$(v'_{3(i-1)+3},o_{i+1})$.
    The right side of \cref{fig:reduction} shows an example of clause gadgets and this concludes the construction.

    Since the reduction can clearly be computed in polynomial time (in~$n+k$), it remains to prove correctness.
    To this end, first assume that the formula~$\phi$ is satisfiable and let~$\beta$ be a satisfying assignment.
    We show that the constructed instance of \cidp{} is a yes-instance.
    We will describe how to construct a path~$P_i$ for each~$i \in [3c-1]$ from~$s_i$ to~$t_i$ such that each vertex is part of at most~$c$ paths.
    
    We first describe how to construct~$P_{\alpha}$ for each~$\alpha \in [c]$.
    Each path~$P_{\alpha}$ starts with the two vertices~$s_{\alpha}$ and~$j_1$.
    Afterwards, they iteratively (in increasing order) pass for each~$i \in [n]$ through all vertices~$u_j$ and~$v_j$ for positions~$j$ where~$x_i$ appears negated if~$\beta$ sets~$x_i$ to true and through all vertices~$u_j$ and~$v_j$ for positions~$j$ where~$x_i$ appears positive if~$\beta$ sets~$x_i$ to false.
    From the last such vertex~$v_j$ they continue to~$j_{i+1}$.
    All~$c$ paths continue from~$j_{n+1}$ to~$o_1$.
    We then pick for each clause~$C_j \in \mathcal{C}$ a position~$r_j$ corresponding to~$C_j$ such that~$C_j$ is satisfied by the variable at position~$r_j$ under assignment~$\beta$.
    Note that such a position always exists as~$\beta$ is a satisfying assignment.
    The~$c$ paths~$P_{\alpha}$ for~$\alpha\in [c]$ then iteratively go through~$o_j$ and~$u'_{r_j}$ and~$v'_{r_j}$ for each~$j \in [m]$.
    From~$o_{m+1}$ they all continue to~$d_{3m}$.
    Next, for each switch gadget~$S_j$ (in decreasing order), each path~$P_{\alpha}$ with~$\alpha \in [c]$ continues with the path~$(d_j,h,v,g,w_\alpha)$ (we omit the index~$j$ from now on for the sake of readability as we will only consider a single switch for now) if the variable~$x_i$ in position~$j$ (i) appears positive in position~$j$ and is set to true by~$\beta$ or (ii) appears negated in position~$j$ and is set to false by~$\beta$.
    Otherwise (if it appears positive and is set to false or appears negated and is set to true), the paths~$P_{\alpha}$ with~$\alpha \in [c]$ continue with the path~$(d,h',v',g',w'_\alpha)$.
    The path~$P_c$ then goes to vertex~$y$ and continues with the path~$(y,\ell_c,\ell_{c-1},\ldots,\ell_1,q_1,p_1,a)$.
    The remaining~$c-1$ paths~$P_{\alpha}$ with~$\alpha \in [c-1]$ continue with the vertex~$z$ and then the path~$(z,r_c,r_{c-1},\ldots,r_1,q_{i+1},p_{i+1},a)$.
    After passing through all switch gadgets in this fashion, each path~$P_{\alpha}$ with~$\alpha \in [c]$ ends with the direct arc from~$a_1$ to~$t_i$.

    We next describe how to construct~$P_i$ for each~$c+1 \leq i \leq 3c-1$.
    We consider three cases, $i < 2c$, $i = 2c$ and~$i>2c$.
    For all three cases, the paths will pass through all switch gadgets in increasing order.
    In the first two cases, they will go from~$b$ to~$e$ and in the last case, they will go from~$x$ to~$y$.
    We will also consider two cases for each switch~$S_j$ based on whether the variable~$x_i$ appearing at position~$j$ is set by~$\beta$ in such a way that it satisfies the clause to which position~$j$ corresponds or not.
    If it does (because~$x_i$ appears positive and is set to true by~$\beta$ or because it appears negated and is set to false by~$\beta$), then we say that switch~$S_j$ is \emph{agreeing} with~$\beta$ and otherwise it is \emph{disagreeing} with~$\beta$.
    For each~$c+1 \leq i < 2c$, the path~$P_i$ consists of the arc~$(s_i,b_1)$.
    Next, for each switch~$S_j$ in increasing order, the path continues with the subpath~$(b,q_1,q_2,\ldots,q_c,f',w'_1,w'_2,\ldots,w'_c,g',u,e)$ if~$S_j$ agrees with~$\beta$ and with the path~$(b,q_1,q_2,\ldots,q_c,f',w_1,w_2,\ldots,w_c,g,u',e)$ if~$S_j$ disagrees with~$\beta$.
    In either case, it goes from~$e_j$ to~$b_{j+1}$ for each~$j \in [3m-1]$ and from~$e_{3m}$ to~$t_i$.
    The path~$P_{2c}$ starts with the arc~$(s_{2c},b_1)$.
    Next for each switch~$S_j$, $P_{2c}$ passes through~$S_j$ from~$b_j$ to~$e_j$ as shown next and then to~$b_{j+1}$ (or to~$t_{2c}$ if~$j=3m$).
    Within~$S_j$, the path is~$(b,p_1,p_2,\ldots,p_c,x,r_c,\ell_c,f,w'_1,w'_2,\ldots,w'_c,g',u,e)$ if~$S_j$ agrees with~$\beta$ and it is~$(b,p_1,p_2,\ldots,p_c,x,r_c,\ell_c,f,w_1,w_2,\ldots,w_c,g,u',e)$ if~$S_j$ disagrees with~$\beta$.
    Finally, for each~$2c+1 \leq i \leq 3c-1$, the path~$P_i$ starts with the arc~$(s_i,x_1)$.
    It then continues through each switch gadget~$S_j$ from~$x_j$ to~$y_j$ as shown next.
    From there it goes to~$x_{j+1}$ (or to~$t_i$ if~$j = 3m$).
    Let~$i' = i - 2c$ for each~$2c+1 \leq i \leq 3c-1$.
    If~$S_j$ agrees with~$\beta$, then~$P_i$ takes the path~$(x,r_{i'},\ell_{i'},f,w_1,w_2,\ldots,w_c,y)$ through~$S_j$.
    If~$S_j$ disagrees with~$\beta$, then~$P_i$ takes the path~$(x,r_{i'},\ell_{i'},f,w'_1,w'_2,\ldots,w'_c,y)$ through~$S_j$.

    Since the path~$P_i$ goes from~$s_i$ to~$t_i$ for each~$i \in [3c-1]$, it only remains to prove that each vertex is part of at most~$c$ of the constructed paths.
    Note that no terminal vertex~$s_i$ or~$t_i$ is used by any constructed path~$P_j$ for any~$i \neq j$.
    Moreover, the paths~$P_i$ with~$c+1 \leq i \leq 3c-1$ only contain the terminal vertices and vertices in switch gadgets.
    Hence, any vertex outside a switch gadget is part of at most~$c$ solution paths.
    It remains to show that each vertex in a switch gadget is part of at most~$c$ constructed paths.
    Consider an arbitrary switch~$S_q$.
    We will show that each vertex in the switch is part of at most~$c$ constructed paths.
    Since the switch~$S_q$ is chosen arbitrarily, this will show that each vertex in the constructed graph is part of at most~$c$ paths and hence conclude the forward direction.
    The vertices~$a,z,$ and~$d$ are only part of paths~$P_i$ with~$i \in [c]$ by construction.
    The vertices~$b,f'$, and~$e$ are only part of paths~$P_i$ with~$c+1 \leq i \leq 2c$.
    The vertices~$x$ and~$f$ are only part of paths~$P_i$ with~$2c \leq i \leq 3c-1$.
    Each vertex~$p_i$ with~$i \geq 2$ is only part of paths~$P_{i-1}$ and~$P_{2c}$ and~$p_1$ is only part of~$P_c$ and~$P_{2c}$.
    Each vertex~$q_i$ with~$q \geq 2$ is only part of paths~$P_{i-2}$ and~$P_j$ with~$c+1 \leq j < 2c$.
    The vertex~$q_1$ is only part of paths~$P_j$ with~$c \leq j < 2c$.
    Each vertex~$\ell_i$ with~$i < c$ is only part of~$P_c$ and~$P_{2c+i}$ and~$\ell_c$ is only part of~$P_c$ and~$P_{2c}$.
    Each vertex~$r_i$ with~$i < c$ is only part of~$P_{2c+i}$ and~$P_j$ with~$1 \leq j < c$.
 Vertex~$r_c$ is only part of~$P_{2c}$ and~$P_j$ with~$1 \leq j < c$.
    Vertex~$y$ is only part of paths~$P_1$ and~$P_i$ with~$2c+1 \leq i \leq 3c-1$.
    For the remaining analysis, we make a case distinction whether the chosen switch agrees with~$\beta$ or not.
    We start with the case where $S_q$ agrees with~$\beta$.
    Then, each vertex~$w_i$ is only part of paths~$P_i$ and~$P_j$ for~$2c+1 \leq j \leq 3c-1$.
    Each vertex~$w'_i$ is only part of paths~$P_j$ with~$c+1 \leq j \leq 2c$.
    Since~$S_q$ agrees with~$\beta$, note that by construction, no solution path passes through~$u$ and~$v$ in the variable gadget of the variable appearing at position~$q$.
    Hence, vertices~$g'$ and~$u$ are only part of paths~$P_i$ with~$c+1 \leq i \leq 2c$.
    Vertices~$g,v,h,u',h',$ and~$v'$ are only part of paths~$P_i$ with~$i \in [c]$ (the former three when initially going through the switch and the latter three potentially when passing through a clause gadget).
    
    We next analyze the case where~$S_q$ disagrees with~$\beta$.
    Then, each vertex~$w_i$ is only part of paths~$P_j$ for~$c+1 \leq j \leq 2c$.
    Each vertex~$w'_i$ is only part of paths~$P_i$ and~$P_j$ with~$2c+1 \leq j \leq 3c-1$.
    Since~$S_q$ disagrees with~$\beta$ and by construction, no solution path passes through~$u'$ and~$v'$ in the clause gadget to which~$q$ corresponds to (as we chose for each clause gadget to pass through a position such that the variable at that position satisfies the clause under assignment~$\beta$).
    Hence, vertices~$g$ and~$u'$ are only part of paths~$P_i$ with~$c+1 \leq i \leq 2c$.
    Vertices~$g',v',h',u,h,$ and~$v$ are only part of paths~$P_i$ with~$i \in [c]$ (the latter three when passing through the variable gadget appearing at position~$q$).
    We have shown for all vertices in all switch gadgets that they are part of at most~$c$ of the constructed solution paths.   
    This concludes the proof of the forward direction.
    
    In the other direction, assume that the constructed instance of \cidp{} has a solution.
    As in the proof of \cref{lem:switch}, it will be convenient to bundle certain paths.
    We say that the paths from~$s_i$ to~$t_i$ are red for all~$i \leq c$ and blue for all~$i > c$.
    It will again be convenient to speak about the reverse of the red paths and hence we will say that red paths start in~$t_i$, end in~$s_i$ and use arcs in the ``wrong'' direction.
    Note that by construction, the second vertex in all red paths is~$a_1$.
    Moreover, the~$c$ (blue) paths from~$s_i$ to~$t_i$ for~$c+1 \leq i \leq 2c$ have~$b_1$ as their second vertex and the remaining~$c-1$ (also blue) paths have~$x_1$ as their second vertex as each~$s_i$ only has one incident arc.
    Hence all requirements of \cref{lem:switch} are satisfied and it implies that all red paths pass through~$d_1$, $c-1$ blue paths leave the first switch~$S_1$ at~$y_1$ and the remaining~$c$ blue paths leave~$S_1$ at~$e_1$.
    Since~$d_1$ only has one incoming arc (from vertex~$a_2$) and~$y_1$ and~$e_1$ each only have one outgoing arc (to~$x_2$ and~$b_2$, respectively), this shows that all requirements of \cref{lem:switch} are also satisfied for~$S_2$.
    The same argument applies inductively for each switch gadget~$S_i$.
    We next focus on the path~$R_1$ from~$t_1$ to~$s_1$ in the assumed solution.
    As shown above, the path has to go through all switch gadgets and pass through~$d_{3m}$.
    The only incoming arc of~$d_{3m}$ comes from~$o_{m+1}$ (note that red paths use arcs in the opposite direction).
    By construction, the path now has to go through all clause and variable gadgets.
    To conclude the proof, we will show that this path encodes a satisfying assignment.

    By construction, for each clause~$C_j$, the path~$R_1$ has to pass through~$v'_{r_j}$ and~$u'_{r_j}$ for some position~$r_j$ corresponding to~$C_j$.
    For each variable~$x_i$, the path~$R_1$ has to pass through either all vertices~$u_{p_j}$ and~$v_{p_j}$ for positions~$p_j$ where~$x_i$ appears positive in~$\phi$ or all vertices~$u_{n_j}$ and~$v_{n_j}$ for positions~$n_j$ where~$x_i$ appears negated in~$\phi$.
    In the former case, we set~$x_i$ to false and in the latter case, we set~$x_i$ to true.
    We will show that this assignment satisfies all clauses.
    To this end, consider any clause~$C_j$ and the variable~$x_i$ appearing at the position~$r_j$ taken by~$R_1$.
    Since \cref{lem:switch} applies to all switch gadgets (in particular to~$S_{r_j}$), the path~$R_1$ cannot simultaneously pass through~$u_{r_j}$ and~$u'_{r_j}$.
    Since it passes through~$u'_{r_j}$, it does not pass through~$u_{r_j}$.
    If~$x_i$ appears positive in position~$r_j$, then~$R_1$ passes through all vertices~$u_{j'}$ for positions~$j'$ where~$x_i$ appears negated and hence we set~$x_i$ to true and~$C_j$ is satisfied by our constructed assignment.
    If~$x_i$ appears negated in position~$r_j$, then~$R_1$ passes through all vertices~$u_{j'}$ for positions~$j'$ where~$x_i$ appears positive in~$\phi$ and hence we set~$x_i$ to false and~$C_j$ is again satisfied by our constructed assignment.
    Since the clause~$C_j$ was chosen arbitrarily, this shows that all clauses are satisfied by the constructed assignment.
    Thus, $\phi$ is satisfiable and this concludes the proof.
\end{proof}

We mention in passing that since \textsc{3-Sat} cannot be solved in~$2^{o(m)}$ time unless the exponential time hypothesis fails \cite{IP01,IPZ01} and the size of the construction in the proof of \cref{thm:nphard} is linear in the number of clauses for any constant~$c$, this also shows that \cidp{} cannot be solved in~$2^{o(n+m)}$ time unless the exponential time hypothesis fails.

\begin{corollary}
    Assuming the ETH, \cidp{} cannot be solved in~$2^{o(n+m)}$~time for any constant~$c \geq 1$.
\end{corollary}

\bibliographystyle{alphaurl} 
\bibliography{bibliography}

\end{document}